\documentclass[letterpaper,11pt]{article}
\usepackage[utf8]{inputenc}

\pdfoutput=1

\usepackage{amsmath, amsthm, amssymb, thm-restate}
\usepackage{algorithmicx}
\usepackage[table,xcdraw]{xcolor}
\usepackage[ruled,vlined,linesnumbered]{algorithm2e}
\usepackage{setspace}
\usepackage{mathtools}
\usepackage[numbers]{natbib}
\usepackage{comment} 
\usepackage{tcolorbox} 
\usepackage{xfrac}
\usepackage{hyperref}
\usepackage{multirow}
\usepackage{caption}
\usepackage{bm}
\usepackage{newfloat}
\usepackage{enumitem}
\usepackage{dblfloatfix} 
\usepackage{wrapfig}

\usepackage{fancyhdr}

\usepackage[margin=1in]{geometry}

\allowdisplaybreaks

\definecolor{mygreen}{RGB}{10,110,230}
\definecolor{myred}{RGB}{10,110,230}

\hypersetup{
     colorlinks=true,
     citecolor= mygreen,
     linkcolor= myred
}

\renewcommand{\epsilon}{\varepsilon}

\DeclareMathOperator{\E}{\ensuremath{\normalfont \textbf{E}}}

\newcommand{\hiddencomment}[1]{}

\newcommand{\mc}[1]{\ensuremath{\mathcal{#1}}}

\newcommand{\K}[0]{\ensuremath{o(n^{6/5})}}
\newcommand{\thelb}[0]{\ensuremath{n^{1.2 - o(1)}}}
\newcommand{\yes}[0]{\ensuremath{\mathsf{YES}}}
\newcommand{\no}[0]{\ensuremath{\mathsf{NO}}}
\newcommand{\yesdist}[0]{\ensuremath{\mathcal{D}_{\yes}}}
\newcommand{\nodist}[0]{\ensuremath{\mathcal{D}_{\no}}}
\newcommand{\dist}[0]{\ensuremath{\mathcal{D}}}

\DeclareMathOperator{\poly}{poly}

\usepackage[noabbrev,nameinlink]{cleveref}
\crefname{lemma}{Lemma}{Lemmas}
\crefname{theorem}{Theorem}{Theorems}
\crefname{property}{Property}{Properties}
\crefname{claim}{Claim}{Claims}
\crefname{result}{Result}{Results}
\crefname{definition}{Definition}{Definitions}
\crefname{observation}{Observation}{Observations}
\crefname{proposition}{Proposition}{Propositions}
\crefname{assumption}{Assumption}{Assumptions}
\crefname{line}{Line}{Lines}
\crefname{figure}{Figure}{Figures}
\creflabelformat{property}{(#1)#2#3}
\crefname{equation}{}{}
\crefname{section}{Section}{Sections}
\crefname{appendix}{Appendix}{Appendices}
\crefname{algCounter}{Algorithm}{Algorithms}
\Crefname{algCounter}{Algorithm}{Algorithms}

\newtheorem{theorem}{Theorem}[section]

\newtheorem{lemma}[theorem]{Lemma}
\newtheorem{proposition}[theorem]{Proposition}
\newtheorem{corollary}[theorem]{Corollary}

\newtheorem{definition}[theorem]{Definition}
\newtheorem{claim}[theorem]{Claim}

\newtheorem{observation}[theorem]{Observation}
\newtheorem{remark}[theorem]{Remark}
\newtheorem*{remark*}{Remark}

\definecolor{mylightgray}{RGB}{230,230,230}

\algnewcommand{\IIf}[2]{\textbf{if} #1 \textbf{then} #2}
\algnewcommand{\EndIIf}{\unskip\ \algorithmicend\ \algorithmicif}

\newenvironment{graytbox}{
\par\addvspace{0.1cm}
\begin{tcolorbox}[width=\textwidth,
                  boxsep=5pt,
                  left=1pt,
                  right=1pt,
                  top=2pt,
                  bottom=2pt,
                  boxrule=0pt,
                  arc=0pt,
                  colback=mylightgray,
                  colframe=black,
                  ]
}{
\end{tcolorbox}
}

\newenvironment{whitetbox}{
\par\addvspace{0.1cm}
\begin{tcolorbox}[width=\textwidth,
                  boxsep=5pt,
                  left=1pt,
                  right=1pt,
                  top=2pt,
                  bottom=2pt,
                  boxrule=1pt,
                  arc=0pt,
                  colframe=black,
                  colback=white
                  ]
}{
\end{tcolorbox}
}

\newenvironment{myproof}{
\vspace{-0.5cm}
\begin{proof}
}{
\end{proof}
}

\newcounter{algCounter}

\makeatletter
\renewcommand{\paragraph}{%
  \@startsection{paragraph}{4}%
  {\z@}{10pt}{-1em}%
  {\normalfont\normalsize\bfseries}%
}
\makeatother

\makeatletter
\patchcmd{\@algocf@start}
  {-1.5em}
  {0pt}
  {}{}
\makeatother

 \title{Sublinear Time Algorithms and Complexity of Approximate Maximum Matching}

\author{
Soheil Behnezhad\\{\em Northeastern University} \and
Mohammad Roghani\\{\em Stanford University} \and
Aviad Rubinstein\thanks{Supported by NSF CCF-1954927 and a David and Lucile Packard Fellowship.}\\{\em Stanford University}
}

\date{}

\begin{document}

\maketitle

\thispagestyle{empty}

\begin{abstract}
     Sublinear time algorithms for approximating maximum matching size have long been studied. Much of the progress over the last two decades on this problem has been on the algorithmic side. For instance, an algorithm of \citet*{behnezhad2021} obtains a 1/2-approximation in $\widetilde{O}(n)$ time for $n$-vertex graphs. A more recent algorithm by \citet*{behnezhadsublinear22} obtains a slightly-better-than-1/2 approximation in $O(n^{1+\epsilon})$ time (for arbitrarily small constant $\varepsilon>0$). On the lower bound side, \citet*{ParnasRon07} showed 15 years ago that obtaining any constant approximation of maximum matching size requires $\Omega(n)$ time. Proving any super-linear in $n$ lower bound, even for $(1-\epsilon)$-approximations, has remained elusive since then.
    
     \smallskip\smallskip
     In this paper, we prove the first super-linear in $n$ lower bound for this problem. We show that at least $n^{1.2 - o(1)}$ queries in the adjacency list model are needed for obtaining a $(\frac{2}{3} + \Omega(1))$-approximation of the maximum matching size. This holds even if the graph is bipartite and is promised to have a matching of size $\Theta(n)$. Our lower bound argument builds on techniques such as {\em correlation decay} that to our knowledge have not been used before in proving sublinear time lower bounds.
     
     \smallskip\smallskip
     We complement our lower bound by presenting two algorithms that run in strongly sublinear time of $n^{2-\Omega(1)}$. The first algorithm achieves a $(\frac{2}{3}-\varepsilon)$-approximation (for any arbitrarily small constant $\varepsilon>0$); this significantly improves prior close-to-1/2 approximations. 
     Our second algorithm obtains an even better approximation factor of $(\frac{2}{3}+\Omega(1))$ for bipartite graphs. This breaks $2/3$-approximation which has been a barrier in various settings of the matching problem, and importantly shows that our $n^{1.2-o(1)}$ time lower bound for $(\frac{2}{3}+\Omega(1))$-approximations cannot be improved all the way to $n^{2-o(1)}$.
\end{abstract}

{
\clearpage
\hypersetup{hidelinks}
\vspace{1cm}
\renewcommand{\baselinestretch}{0.1}
\setcounter{tocdepth}{2}
\tableofcontents{}
\thispagestyle{empty}
\clearpage
}

\setcounter{page}{1}
\clearpage

\section{Introduction}

We study algorithms for estimating the size of {\em maximum matching}. Recall that a {\em matching} is a set of edges no two of which share an endpoint, and a {\em maximum matching} is a matching of the largest size. This problem has been studied for decades. There are traditional algorithms that can obtain an arbitrary good approximation of the solution in time linear in the input size (see e.g. the celebrated algorithm of Hopcroft and Karp \cite{HopcroftK73}). A natural question, that has received significant attention over the past two decades, is whether it is possible to estimate the size of maximum matching in {\em sublinear time} in the input size \cite{ParnasRon07,NguyenOnakFOCS08,YoshidaYISTOC09,OnakSODA12,KapralovSODA20,ChenICALP20,behnezhad2021,behnezhadsublinear22}.

Since any sublinear time algorithm can read only a small fraction of the input, it is important to specify how the input is represented. Our focus in this work is particularly on the standard {\em adjacency list} model. Here, the algorithm can specify a vertex $v$ of its choice and any integer $i$, and in response receives the $i$-th neighbor of $v$ stored in an arbitrarily ordered list, or receives NULL if $v$ has less than $i$ neighbors.\footnote{Another common model is the {\em adjacency matrix} model where the algorithm can specify a pair of vertices $(u, v)$ and in response gets whether they are adjacent or not.} Given a graph $G$ represented in this model, the goal is to provide an estimate $\widetilde{\mu}(G)$ of the size $\mu(G)$ of the maximum matching of $G$ such that $\widetilde{\mu}(G)$ is ``close'' to $\mu(G)$ while the total number of queries made to the graph is much smaller than the size of the graph.

\subsection{Known Bounds}

It is not hard to see that to provide an exact estimate satisfying $\widetilde{\mu}(G) = \mu(G)$, one has to make $\Omega(n^2)$ queries to the graph, even if the algorithm is randomized (with constant success probability). Thus, any sublinear time algorithm must resort to approximations. Indeed, multiple algorithms have been devised that run in sublinear time and approximate maximum matching size. Earlier results, pioneered in the works of \citet*{ParnasRon07,NguyenOnakFOCS08} and \citet*{YoshidaYISTOC09} focused mainly on bounded-degree graphs. The more recent works of \cite{ChenICALP20,KapralovSODA20,behnezhad2021,behnezhadsublinear22} run in sublinear time even on general $n$-vertex graph. For instance, a result of \citet*{behnezhad2021} shows an (almost) 1/2-approximation in the adjacency list model can be obtained in $\widetilde{O}(n)$ time. More recently, \citet*{behnezhadsublinear22} broke the half-approximation barrier, obtaining a (slightly) better $(\frac{1}{2}+\Omega(1))$-approximation in $O(n^{1+\epsilon})$-time. With the 1/2-approximation now broken, a natural next step, posed explicitly in \cite{behnezhadsublinear22} as an open problem, is determining the best approximation achievable in $n^{2-\Omega(1)}$ time.

The situation on the lower bound side is much less understood. The only known lower bound is that $\Omega(n)$ time is needed for obtaining any constant approximation of maximum matching, which was proved 15 years ago by \citet*{ParnasRon07}. While this lower bound was (nearly) matched by \cite{behnezhad2021} for 1/2-approximations, it is not known whether it is optimal or can be improved for better approximations. In particular, the current state of affairs leave it possible to obtain even a $(1-\epsilon)$-approximation, for any fixed $\epsilon > 0$, in just $O(n)$ time. Not only such a result would be amazing on its own, but as we will later discuss, it will have deep consequences in the study of dynamic graphs. It is also worth noting that in their beautiful work, \citet*{YoshidaYISTOC09} showed existence of an $O(n) + \Delta^{O(1/\epsilon^2)}$ time algorithm that obtains a $(1-\epsilon)$-approximation\footnote{More precisely, \cite{YoshidaYISTOC09} give a multiplicative-additive $(1-\epsilon, o(n))$ approximation in $\Delta^{O_\epsilon(1)}$ time.  The claimed bound follows by slightly tweaking their algorithm using techniques developed in \cite{behnezhad2021} for multiplicative approximations.}. While this is not a sublinear time algorithm for the full range of $\Delta$, it runs in $O(n)$ time for $\Delta = n^{O(\epsilon^2)}$. This shows that any potential $\omega(n)$ time lower bound must be proved on graphs of large degree.

\subsection{Our Results}

\paragraph{Improved Lower Bound:} In this work, we give the first super-linear in $n$ lower bound for the sublinear matching problem. We show that not only $(1-\epsilon)$-approximations are not achievable in $O(n)$ time, but in fact any better than $2/3$-approximation requires at least $n^{1.2-o(1)}$ time.

\begin{graytbox}
    \begin{theorem}\label{thm:lb}
        For any fixed $\alpha > 0$, any (possibly randomized) algorithm obtaining a $(2/3+\alpha)$-approximation of the size of maximum needs to make at least $n^{1.2 - o(1)}$ adjacency list queries to the graph. This holds even if the graph is bipartite and has a matching of size $\Theta(n)$.
    \end{theorem}
\end{graytbox}

Our proof of \cref{thm:lb} relies crucially on {\em correlation decay}. To our knowledge, this is the first application of correlation decay in proving sublinear time lower bounds. We give an informal overview of our techniques in proving \cref{thm:lb} in \cref{sec:techniques}, and discuss why correlation decay is extremely helpful for us. The formal proof of \cref{thm:lb} is then presented in \cref{sec:lb}.

We note that the essence of the $\Omega(n)$ lower bound of \citet*{ParnasRon07} is that $o(n)$ queries are not enough to even see all neighbors of a single vertex. Using this, \cite{ParnasRon07} constructs an input distribution where no $o(n)$ time algorithm can see any edge of any (approximately) optimal matching. Indeed one key challenge that our lower bound of \cref{thm:lb} overcomes is to show that even though there are, say, $O(n^{1.1})$ time algorithms that ``see'' as many as $n^{\Omega(1)}$ edges of an optimal matching, they are still unable to obtain a $(2/3+\Omega(1))$-approximation.

\paragraph{Improved Upper Bounds:} We also present two algorithms that run in strictly sublinear time of $n^{2-\Omega(1)}$. Our first result is an algorithm that works on general (i.e., not necessarily bipartite) graphs and obtains an (almost) 2/3-approximation. This significantly improves prior close-to-1/2 approximations of \cite{behnezhad2021,behnezhadsublinear22}.

\begin{graytbox}
    \begin{theorem}\label{thm:2/3}
        For any fixed $\epsilon > 0$, there are algorithms for approximating the maximum matching size of any (general) $n$-vertex graph that take $n^{2-\Omega_\epsilon(1)}$ time and obtain
        \begin{itemize}[itemsep=0pt, topsep=5pt]
            \item a multiplicative $(2/3-\epsilon)$-approximation in the adjacency list model, and 
            \item a multiplicative-additive $(2/3-\epsilon, o(n))$-approximation in the adjacency matrix model.
        \end{itemize}
    \end{theorem}
\end{graytbox}

\begin{remark*} A subquadratic time multiplicative $O(1)$-approximation is impossible in the adjacency matrix model since even distinguishing an empty graph from one including only one random edge requires $\Omega(n^2)$ adjacency matrix queries. Therefore, a multiplicative-additive approximation (defined formally in \cref{sec:prelim}) is all we can hope for under adjacency matrix queries.
\end{remark*}

Although \cref{thm:2/3} significantly improves prior approximations and comes close to the 2/3 barrier, but it does not break it. As such, one may still wonder whether the $n^{1.2-o(1)}$ lower bound of \cref{thm:lb} for $(2/3+\Omega(1))$-approximations can be improved to $n^{2-o(1)}$. Our next result rules this possibility out, and gives a subquadratic time algorithm that indeed breaks $2/3$-approximation provided that the input graph is bipartite.

\begin{graytbox}
    \begin{theorem}\label{thm:beating-2/3}
        There are algorithms for approximating the maximum matching size of any bipartite $n$-vertex graph that take $n^{2-\Omega(1)}$ time and obtain
        \begin{itemize}[itemsep=0pt, topsep=5pt]
            \item a multiplicative $(2/3+\Omega(1))$-approximation in the adjacency list model, and 
            \item a multiplicative-additive $(2/3+\Omega(1), o(n))$-approximation in the adjacency matrix model.
        \end{itemize}
    \end{theorem}
\end{graytbox}

\cref{thm:lb,thm:beating-2/3} together imply a rather surprising bound on the complexity of $(2/3+\Omega(1))$-approximating maximum matching size for bipartite graphs: They show that the right time-complexity is $n^{1+c}$ for some $c$ that is strictly smaller than 1 but no smaller than $.2$.

\paragraph{Independent Work:} In a concurrent and independent work \citet*{BKS22} gave $n^{2-\Omega(1)}$ time algorithms for obtaining (almost) 2/3-approximation in both adjacency list and adjacency matrix models similar to our \cref{thm:2/3}. We note that the lower and upper bounds of \cref{thm:lb,thm:beating-2/3} are unique to our paper.

\subsection{Further Related Works and Implications}

\paragraph{The 2/3-Barrier for Approximating Matchings:} Prior to our work, it was shown in a beautiful paper of \citet*{GoelKK12} that obtaining a better than 2/3-approximation of maximum matching in the one-way communication model requires $n^{1+\Omega(1/\log \log n)} \gg n\poly(\log n)$ communication from Alice to Bob. This was, to our knowledge, the first evidence that going beyond 2/3-approximation for maximum matching is ``hard'' in a certain model. Our \cref{thm:lb} shows that this barrier extends to the sublinear model. While the two models are completely disjoint and the constructions are different, we note that unlike \cref{thm:lb}, the $n^{1+\Omega(1/\log \log n)}$ lower bound of \cite{GoelKK12} is still $n^{1+o(1)}$. Obtaining an $n^{1+\Omega(1)}$ lower bound for approximating matchings in the one-way communication model remains open (and, in fact, impossible short of a breakthrough in combinatorics --- see \cite{GoelKK12}).

\paragraph{Implications for Dynamic Algorithms:} In the fully dynamic maximum matching problem, we have a graph that is subject to both edge insertions and deletions. The goal is to maintain a good approximation of maximum matching by spending a small time per update. This is a very well studied problem for which several update-time/approximation trade-offs have been shown. Recent works of \citet*{behnezhad2022dynamic} and \citet*{bhattacharyaKSW-SODA23} established a new connection between this problem and (static) sublinear time algorithms. In particular, they showed that any $T$ time algorithm for $\alpha$-approximating maximum matching size leads to an algorithm that maintains an (almost) $\alpha$-approximation in $\widetilde{O}(T/n)$ time per update in the dynamic setting. This, in particular, motivates the study of sublinear time maximum matching algorithms that run in $\widetilde{O}(n)$ time as they would lead to $\widetilde{O}(1)$ update-time algorithms which are the holy grail of dynamic algorithms. Our \cref{thm:lb} shows that, at least in this framework, better than a $2/3$-approximation requires $n^{0.2-o(1)}$ update-time. We note that proving such a lower bound for all algorithms (i.e., not necessarily in this framework) remains an important open problem (see \cite{AbboudW14,HenzingerKNS15} and the references therein).

\section{Our Techniques}\label{sec:techniques}

\subsection{The Lower Bound of \cref{thm:lb}}

\paragraph{An insightful, but broken, input distribution:} Let us start with the lower bound. Consider the following input distribution, illustrated in the figure below. While we emphasize that this is not the final input distribution that we prove the lower bound with, it provides the right intuition and also highlights some of the challenges that arise in proving the lower bound.

 There are four sets of vertices, $A, B, S, T$ with $|A| = |B| = |S| = N$ and $T = \epsilon N$. The $T$ vertices are ``dummy'' vertices, they are adjacent to every other vertex in $A \cup B \cup S$. While this increases the degree of every vertex in $A \cup B \cup S$ to at least $\epsilon N$, the $T$ vertices cannot contribute much to the output since $|T| = \epsilon N$. The important edges, that determine the output are among the rest of the vertices. In particular, there is always a perfect matching between the $B$ vertices and the $S$ vertices. Additionally, there is a random Erd\"os-Renyi graph between $A$ and $B$ for a suitable expected degree $d = N^{\Theta(1)}$. Now in the \yes{} distribution, we also add a perfect matching among the $A$ vertices (the red matching) but in the \no{} distribution, we do not add this matching. Observe that in the \yes{} case, we can match all of $B$ to $S$ and all of $A$ together, obtaining a matching of size at least $3N$. However, in the \no{} case, it can be verified that the maximum matching is only at most $(2+\epsilon)N$. Thus, any algorithm that beats 2/3-approximation by a margin more than $\epsilon$, should be able to distinguish whether we are in the \yes{} distribution or in the \no{} distribution.

\begin{figure}[h]
    \centering
    \includegraphics[scale=0.5]{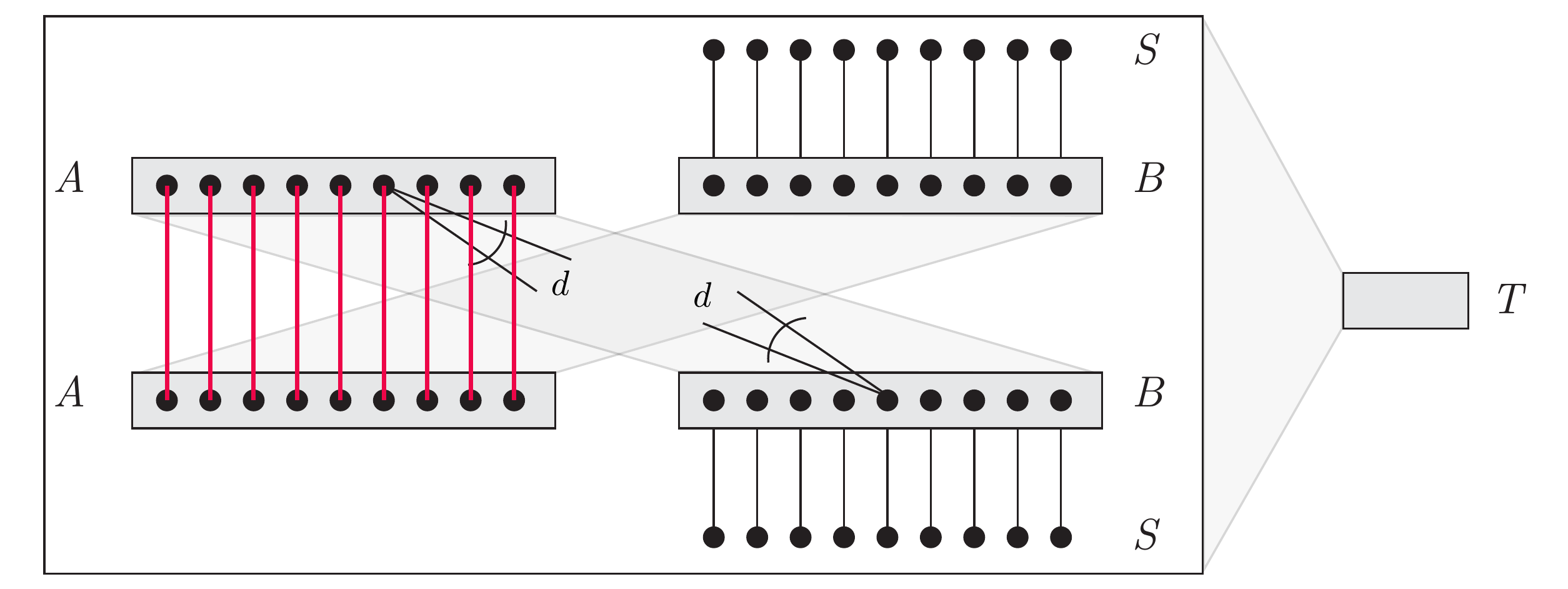}
\end{figure}

Suppose now that we give away which vertices belong to $S$ and $T$ for free with no queries, but keep it secret whether a vertex $v \in A \cup B$ belongs to $A$ or $B$. To examine whether a vertex $v$ belongs to $A$ or $B$, the naive approach is to go over all of its neighbors, and see whether we see an $S$ neighbor or not. But because the degree of each vertex is $\Omega(\epsilon N)$ due to the edges to $T$, this requires at least $\Omega(\epsilon N) = \Omega(N)$ time. But this is not enough. To separate the two cases, one has to actually determine if there are any edges among the $A$ vertices or not. Here, the naive approach is to first find a vertex $v \in A$ which can be done in $O(N)$ time, and then go over all neighbors $u$ of $v$ in $A \cup B$, and examine whether $u$ belongs to $A$. Since $v$ has $d = N^{\Omega(1)}$ neighbors in $A \cup B$, this naive approach takes $N^{1+\Omega(1)}$ time. We would like to argue that this naive approach is indeed the best possible, and that any algorithm distinguishing the two distributions must make $N^{1+\Omega(1)} = n^{1+\Omega(1)}$ queries to the graph. Unfortunately, this is not the case with the distribution above as we describe next. 

Consider the following algorithm. We first start by finding an $A$ vertex $v$ in $O(n)$ time. Then we go over the neighbors of $v$ in a random order until we find the first neighbor $u$ that belongs to $A \cup B$. Note that this takes $O(n/d)$ time. We then do the same for $u$, finding a random neighbor to another vertex in $A \cup B$. Since each step of this random walk takes $O(n/d)$ time, we can continue it for $2d$ steps in $O(n)$ total time. Let $w$ be the last vertex of the walk. We now examine whether $w$ belongs to $A$ or $B$ in $O(n)$ time. We argue that by doing so, there is a constant probability that we can distinguish the \yes{} case from the \no{} case. To see this, observe that in the \no{} distribution, every $A$ vertex goes to a $B$ vertex and every $B$ vertex goes to an $A$ vertex with probability one. As such, since the walk continues for an even number of steps, the last vertex $w$ must belong to $A$ with probability 1. But in the \yes{} distribution, there is a constant probability that we go through an $A$-$A$ edge exactly once. In this case, the last vertex $w$ of the random walk must belong to $B$. Repeating this process a constant number of times allows us to distinguish the two distributions with probability $0.99$ in merely $O(n)$ time!

\paragraph{Correlation Decay to the Rescue:} To get around this challenge, we add $\epsilon d$ edges also among the $B$ vertices (while modifying the size of $A$ and $B$ slightly to ensure that the degrees in $A$ and $B$ do not reveal any information --- see \cref{fig:dist_common}). Observe that this completely destroys the algorithm above. In particular, it no longer holds that any $B$ vertex goes to an $A$ vertex with probability $1$. Rather, there is a probability $\epsilon$ that we go from $B$ to $B$. Hence, intuitively, whether the last vertex $w$ of the random walk belongs to $A$ or $B$ does not immediately reveal any information about whether an $A$-$A$ edge was seen or not. We show that indeed, this can be turned into a formal lower bound against all algorithms via \underline{correlation decay}. First, we show that for suitable $d$, the queried subgraph will only be a tree. We then show that conditioning on the queries conducted far away from a vertex $v$, the probability of $v$ being an $A$ or $B$ vertex will not change, and use this to prove that the algorithm cannot distinguish the \yes{} and \no{} distributions.

\subsection{The Algorithms of \cref{thm:2/3,thm:beating-2/3}}

We now turn to provide an overview of the techniques in obtaining the (almost) 2/3-approximation of \cref{thm:2/3}, and the $(2/3 + \Omega(1))$-approximation of \cref{thm:beating-2/3} for bipartite graphs.

\paragraph{An (Almost) 2/3-Approximation via EDCS:} The {\em edge-degree constrained subgraph} (EDCS) introduced by \citet*{bernsteinstein2015,bersteinsteingeneral} has been a powerful tool in obtaining an (almost) $2/3$-approximation of maximum matching in various settings including dynamic algorithms \cite{bernsteinstein2015,bersteinsteingeneral}, communication complexity \cite{AssadiB19}, stochastic matchings \cite{AssadiB19}, and (random order) streaming \cite{bernsteinstreaming2020,AssadiB21}. In this work, we use it for the first time in the sublinear time model. In particular, both \cref{thm:2/3,thm:beating-2/3} build on EDCS. 

For a parameter $\beta$ (think of it as a large constant), a subgraph $H$ of a graph $G$ is a $\beta$-EDCS of $G$ if $(P1)$ all edges $(u, v)$ in $H$ satisfy $\deg_H(u) + \deg_H(v) \leq \beta$ and $(P2)$ all edges $(u, v) \in G \setminus H$ satisfy $\deg_H(v) + \deg_H(v) \geq (1-\epsilon)\beta$. The main property of EDCS, proved first by \cite{bernsteinstein2015,bersteinsteingeneral} and further strengethened in \cite{AssadiB19,BehnezhadEdmondsGallai}, is that for any $\beta = \Omega(1/\epsilon)$, any $\beta$-EDCS includes a $(1-O(\epsilon))2/3$ approximate maximum matching of its base graph $G$. 

It is not hard to see that {\em finding} any $\beta$-EDCS of the whole input graph $G$ requires $\Omega(n^2)$ queries.\footnote{In fact, finding the edges of any constant approximate matching requires $\Omega(n^2)$ time. Our focus throughout this paper is on estimating only the size of the maximum matching.} Therefore, instead, we first sub-sample some $p=1/n^{\delta}$ fraction of the edges (or pairs in the adjacency matrix model) of $G$ for some fixed $\delta > 0$ in $O(n^2 p) = n^{2-\Omega(1)}$ time and construct an EDCS $H$ over those edges. Building on an approach of \citet*{bernsteinstreaming2020} in the random-order streaming setting, this can be done in a way such that at most $\widetilde{O}(\mu(G) / p ) = n^{2-\Omega(1)}$ edges in the whole graph $G$ remain {\em underfull}, i.e., those that do not satisfy property $(P2)$. The union of the set $U$ of underfull edges and the EDCS $H$ can be shown to include an (almost) 2/3-approximation. The next challenge is that we do not have the set $U$. However, to estimate the maximum matching of $H \cup U$, it suffices to design an oracle that upon querying a vertex $v$, determines whether it belongs to an approximately optimal maximum matching of $H \cup U$ and run this oracle on a few randomly sampled vertices. To do so, we build on a local computation algorithm (LCA) of \citet*{levironitt} (which itself builds on a result of \citet*{YoshidaYISTOC09}) that takes $\poly_\epsilon(\Delta)$ time to return if a given vertex $v$ belongs to some $(1-\epsilon)$-approximate matching of its input graph, where here $\Delta$ is the maximum degree. Modifying $H \cup U$ by getting rid of its high-degree vertices, we show the algorithm of \cite{levironitt} can be used to $(1-\epsilon)$-approximate the size of the maximum matching of $H \cup U$ in $O(n^{1+\delta/\epsilon^2})$ time. Picking $\delta$ sufficiently small, we arrive at an algorithm that obtains a $(2/3-\epsilon)$-approximation in $n^{2-\Omega_\epsilon(1)}$ time.

\paragraph{Going Beyond $2/3$-Approximations:} It is known that the (almost) $2/3$-approximation analysis for EDCS is tight. That is, there are instances on which the EDCS does not include a better than $2/3$-approximation. However, a characterization of such tight instances of EDCS was recently given in the work of \citet*{behnezhad2022dynamic} that we use in beating $2/3$-approximation. Consider a $\beta$-EDCS $H$, for sufficiently large constant $\beta$, that does not include a strictly better than 2/3 approximation of $(2/3+\delta)$-approximation for some small $\delta > 0$. The characterization divides the vertices into two subsets $V_{mid}$ and $V_{low}$ depending on their degrees in $H$. Then shows that there must be an (almost) 2/3-approximate matching $M$ in the induced bipartite subgraph $G[V_{mid}, V_{low}]$ that is far from being maximal. Namely, it leaves an (almost) 1/3-approximate matching in $G[V \setminus M]$ that can be directly added to $M$. Interestingly, this characterization coincides with our lower bound construction! The set $V_{mid}$, here, corresponds to the set $A \cup B$ in the lower bound and the set $V_{low}$ corresponds to $S$. This essentially reduces the problem to showing that our lower bound construction can indeed be solved in $n^{2-\Omega(1)}$ time for all choices of the degree $d$. 

Here we show why our lower bound construction fundamentally cannot lead to a better than $n\sqrt{n}$ time lower bound. Indeed the ideas behind our algorithm for \cref{thm:beating-2/3} are very close. First, suppose that $d < \sqrt{n}$. In this case, we can random sample a vertex in $A$ and examine the $A$/$B$ value of each of its $d$ neighbors in total $O(nd) = O(n\sqrt{n})$ time to see if it has any $A$ neighbors, thereby solving the instance. So let us suppose that $d \geq \sqrt{n}$. Now in this case, we can first sample a vertex $v$ in $A$ and list all of its $A \cup B$ neighbors in $O(n)$ time. Since the vast majority (i.e., all but one) of the neighbors of $v$ go to $B$, we have essentially found $d$ vertices in $B$ in $O(n)$ time. Repeating this $\Theta(n/d)$ times, we find all of the $B$ vertices in total $O(n^2/d) = O(n\sqrt{n})$ time. We note that the final running time of our algorithm in \cref{thm:beating-2/3} is only $n^{2-\Omega_\epsilon(1)}$ and not $O(n\sqrt{n})$ since the bottleneck is in finding the EDCS and estimating the maximum matching size of $H \cup U$ as discussed above. For a more detailed high level overview of the algorithm, see \cref{sec:beating}.

\section{Preliminaries}\label{sec:prelim}

\paragraph{Notation:} Throughout this paper, we let $G = (V, E)$ to denote the input graph. We use $n$ and $m$ to denote the number of vertices and number of edges of the input graph $G$. Also, we use $\Delta$ to show the maximum degree, $\bar{d}$ to show the average degree, and $\mu(G)$ to show the size of the maximum matching of graph $G$.

Throughout this paper, for $\alpha \in (0, 1]$ and $\delta \in [0, n]$, we say $\widetilde{\mu}(G)$ is an approximation of $\mu(G)$ up to a multiplicative-additive error of ($\alpha, \delta$), if $\alpha \mu(G) - \delta \leq \widetilde{\mu}(G) \leq \mu(G)$. Also, we use $\widetilde{O}(\cdot)$ to hide $\poly \log n$ factors.

\paragraph{Problem Definition:} Given a graph $G$, we are interested in estimating the size of the maximum matching of $G$. The graph is given in one of the following representations:

\begin{itemize}
    \item \textit{Adjacency Matrix:} In this model, the algorithm can query a pair of vertices $(u,v)$. The answer is 1 if there is an edge between $u$ and $v$, and 0 otherwise.

    \item \textit{Adjacency List:} In this model, a list of neighbors of each vertex is stored in arbitrary order in a list. The algorithm can query the $i$'th element in the list of neighbors of a vertex. The answer is empty if the $i$'th element does not exist.
\end{itemize}

\paragraph{Random Greedy Maximal Matching:} Let $G = (V, E)$ be the input graph and $\pi$ be a random permutation over all possible permutation of edges $E$. A random greedy maximal matching can be obtained by iterating over $E$ with respect to the permutation $\pi$ and constructing a matching by adding edges one by one if they do not violate the matching constraint.

\paragraph{Probabilistic Tools:} In this paper, we use the following standard concentration inequalities.

\begin{proposition}[Chernoff Bound]\label{prop:chernoff}
    Let $X_1, X_2, \ldots, X_n$ be independent Bernoulli random variables, and let $X = \sum_{i=1}^{n} X_i$. For any $t > 0$, 
    $\Pr[|X - \E[X]| \geq t] \leq 2\exp\left(-\frac{t^2}{3\E[X]}\right).$
\end{proposition}

\begin{proposition}[Hoeffding’s Inequality]\label{prop:hoeffding}
    Let $X_1, X_2, \ldots, X_n$ be independent random variables such that $a \leq X_i \leq b$. Let $\overline{X} = (\sum_{i=1}^n X_i)/n$. For any $t > 0$, 
    $
    \Pr[|\overline{X} - \E[\overline{X}]| \geq t] \leq 2\exp\left(-\frac{2nt}{(b-a)^2}\right).
    $
\end{proposition}

\paragraph{Graph Theory:} We use the following classic theorem by K\"{o}nig’s.

\begin{proposition}[K\"{o}nig’s Theorem]
    In any bipartite graph, the size of maximum matching is equal to the size of the minimum vertex cover.
\end{proposition}

\section{The Lower Bound}\label{sec:lb}

In this section, we prove the lower bound of \cref{thm:lb}.

To prove \cref{thm:lb}, we first give an input distribution in \cref{sec:input-dist}. We then prove that any {\em deterministic} algorithm which with probability at least .51 (taken over the randomization of the input distribution) returns a $(2/3+\alpha)$-approximation of the size of maximum matching of the input must make at least \thelb{} queries to the graph. By the `easy' direction of Yao's minimax theorem \cite{Yao77}, this also implies that any randomized algorithm with success probability .51 over worst-case inputs must make at least $\thelb$ queries to obtain a $(2/3+\alpha)$-approximation.

\subsection{The Input Distribution}\label{sec:input-dist}

We start by formalizing the input distribution. Let $\epsilon := \alpha/100$, where $\alpha$ is as in \cref{thm:lb}, let $N$ be a parameter which controls the number of vertices and let $d = N^{1/5}$. In our construction, we assume $N$, $\epsilon N$, $d$, and $\epsilon d$ are all integers; note that this holds for infinitely many choices of $N$. For $n := 6N$, we construct an $n$-vertex bipartite graph $G=(V, U, E)$. We first categorize the vertices in $G$, then specify its edges.

\paragraph{The vertex set:}  The vertex set $V$ is composed of four distinct subsets $A_V, B_V, S_V, T_V$ and similarly $U$ is composed of $A_U, B_U, S_U, T_U$. Throughout, we may write $A, B, S, T$ to respectively refer to sets $A_V \cup A_U$, $B_V \cup B_U$, $S_V \cup S_U$, and $T_V \cup T_U$. In our construction, we will have
$$
    |A_V| = |A_U| = (1-\epsilon) N, \qquad |B_V| = |B_U| = |S_V| = |S_U| = N, \qquad |T_U| = |T_V| = \epsilon N.
$$
We randomly permute all the vertices in $A_V$ and use $v_i(A_V)$ to denote the $i$-th vertex of $A_V$ in this permutation. We do the same for $A_U, B_V, B_U, S_V, S_U$.

\paragraph{The edge set:} For the edge-set $E$, we give two distributions: in distribution \yesdist{} the maximum matching of $G$ is sufficiently large, and in distribution \nodist{} the maximum matching of $G$ is sufficiently small. The final input distribution $\dist := (\frac{1}{2} \yesdist + \frac{1}{2} \nodist)$ draws its input from \yesdist{} with probability 1/2 and from \nodist{} with probability 1/2. The following edges are common in both disributions \yesdist{} and \nodist{}:
\begin{itemize}
    \item All vertices in $T_U$ (resp. $T_V$) are adjacent to all of $A_V, B_V, S_V$ (resp. $A_U, B_U, S_U$).
    \item For any $i \in [N]$, we add edges $(v_i(B_U), v_i(S_V))$ and $(v_i(B_V), v_i(S_U))$ to the graph. In words, there are perfect matchings between $B_U$ and $S_V$ and between $B_V$ and $S_U$.
    \item Let $\mathcal{R}(B_V, B_U)$ be the set of all $(\epsilon d - 1)$-regular graphs $H$ between $B_V$ and $B_U$ such that for all $i \in [N]$, $(v_i(B_V), v_i(B_U)) \not\in H$. We draw one regular graph $R(B_V, B_U)$ from $\mathcal{D}(B_V, B_U)$ uniformly at random and add its edges to $G$.
    \item Let $\mathcal{R}(A_V, B_U)$ be the set of all graphs $H$ between $A_V$ and $B_U$ such that $\deg_H(v) = d$ for all $v \in A_V$, $\deg_H(u) = (1-\epsilon)d$ for all $u \in B_U$, and for all $i \in [(1-\epsilon)N]$, $(v_i(A_V), v_i(B_U)) \not\in H$. We draw one regular graph $R(A_V, B_U)$ from $\mathcal{R}(A_V, B_U)$ uniformly at random and add its edges to $G$.
    \item Let $\mathcal{R}(B_V, A_U)$ be the set of all graphs $H$ between $B_V$ and $A_U$ such that $\deg_H(v) = d$ for all $v \in A_U$, $\deg_H(u) = (1-\epsilon)d$ for all $u \in B_V$, and for all $i \in [(1-\epsilon)N]$, $(v_i(B_V), v_i(A_U)) \not\in H$. We draw one regular graph $R(B_V, A_U)$ from $\mathcal{R}(B_V, A_U)$ uniformly at random and add its edges to $G$.
    \item For all $i \in \{(1-\epsilon)N+1, \ldots, N\}$ we add edge $(v_i(B_V), v_i(B_U))$ to $G$.
\end{itemize}
The following edges are specific to \yesdist{} and \nodist{} respectively:
\begin{itemize}
    \item In \yesdist{}, we additionally add edges $(v_i(A_V), v_i(A_U))$ and $(v_i(B_V), v_i(B_U))$ for all $i \in [(1-\epsilon)N]$.
    \item In \nodist{}, we additionally add edges $(v_i(A_V), v_i(B_U))$ and $(v_i(B_V), v_i(A_U))$ for all $i \in [(1-\epsilon)N]$.
\end{itemize}

This concludes the construction of graph $G$. We also emphasize that the adjacency list of each vertex includes its neighbors in a random order chosen uniformly and independently. 

Finally, we note that we give away the bipartition $V, U$ of the graph for free. Additionally, we also give away which of the sets $S, T, A \cup B$ any vertex belongs to. What is crucially hidden from the algorithm, however, is whether a vertex $v \in A \cup B$ belongs to $A$ or $B$.

\begin{figure}
    \centering
    \includegraphics[scale=0.5]{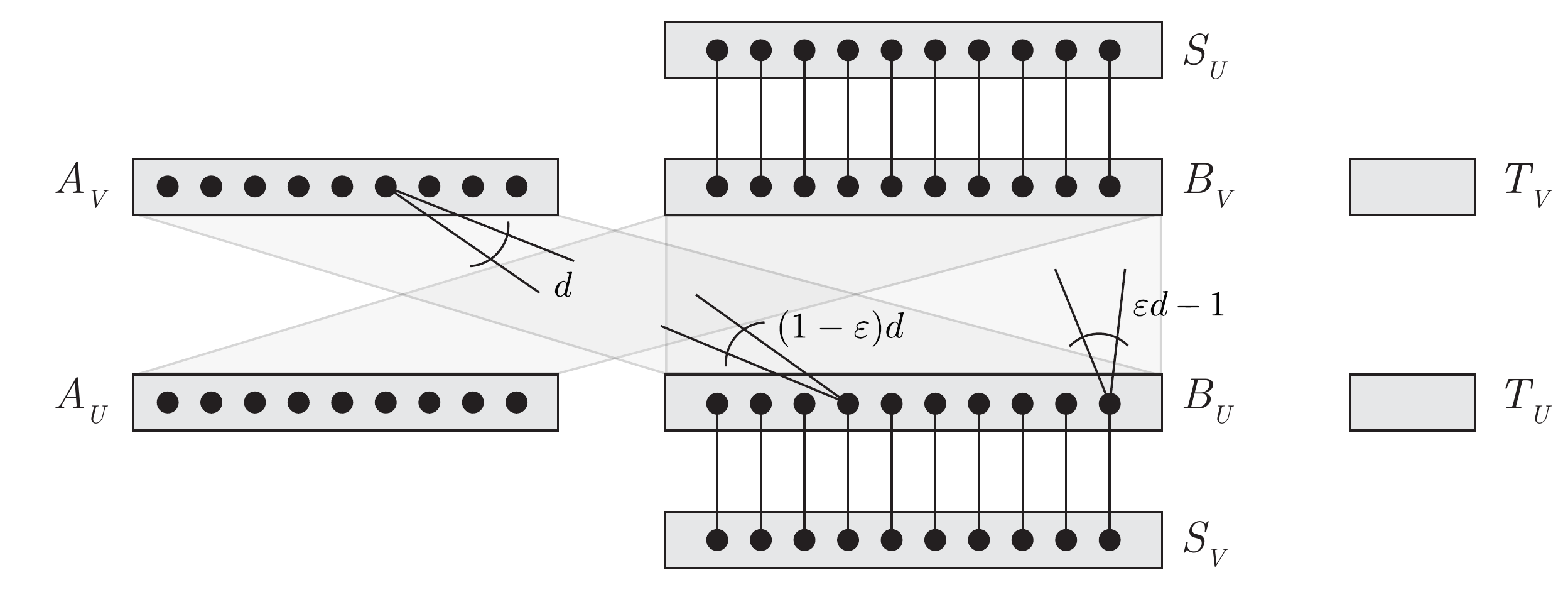}
    \caption{The common edges in both \yesdist{} and \nodist{} distributions. For simplicity, we have not illustrated the edges of $T_V$ and $T_U$ (which are adjacent to all vertices on the opposite side). See \cref{fig:dist_yesno} for edges specific to the two distributions.}
    \label{fig:dist_common}
\end{figure}

\begin{figure}
    \centering
    \includegraphics[scale=0.4]{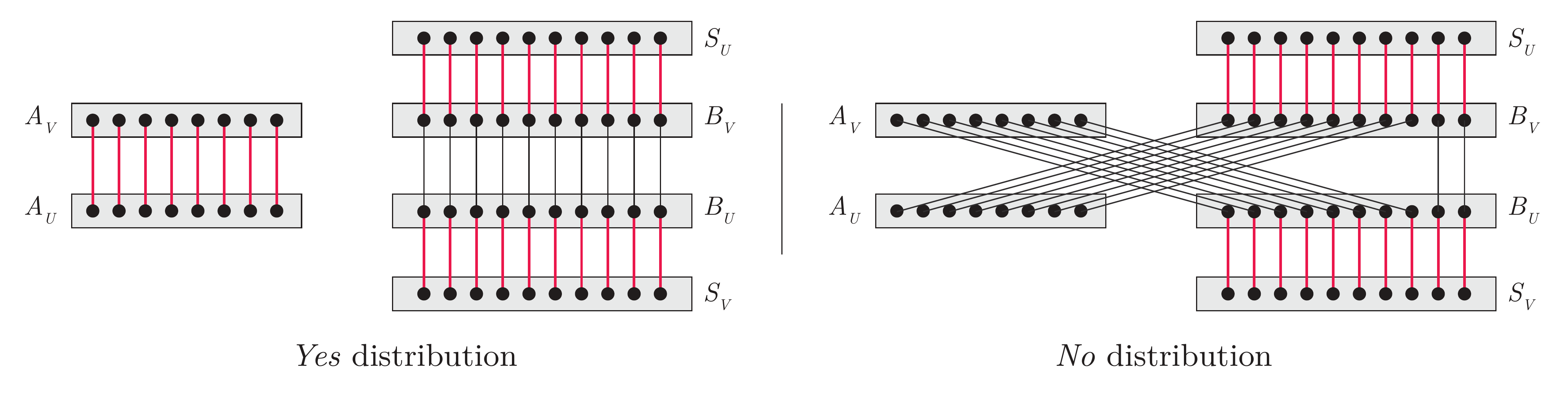}
    \caption{In addition to the common edges in distributions \yesdist{} and \nodist{} that were illustrated in \cref{fig:dist_common}, we have a special perfect matching between vertices $A_V \cup B_V$ and $A_U \cup B_U$. In distribution \yesdist{}, this perfect matching matches all of $A_V$ to $A_U$. But in distribution \nodist{}, none of the edges of this matching go from $A_V$ to $A_U$. This ensures that the maximum matching of $G$ in the \yes{} case is (almost) $1.5$ times that of $G$ in the \no{} case.}
    \label{fig:dist_yesno}
\end{figure}

\subsection{Basic Properties of the Input Distribution}

The following bounds on vertex degrees follows immediately from the construction.

\begin{observation}\label{obs:degrees}
    For any graph $G$ drawn from $\yesdist$ or $\nodist$ with probability 1 it holds that:
    \begin{enumerate}
        \item $\deg(v) = \epsilon N + d + 1$ for all $v \in A_V, A_U, B_V, B_U$.
        \item $\deg(v) = \epsilon N + 1$ for all $v \in S_V, S_U$.
        \item $\deg(v) = (3-\epsilon)N$ for all $v \in T_V, T_U$.
    \end{enumerate}
\end{observation}

Since for any vertex $v$, we know which of the sets $S_V, S_U, T_V, T_U, A_V \cup B_V, A_U \cup B_U$ it belongs to, \cref{obs:degrees} implies we know the degree of every vertex in the graph before making any queries. Thus, there is no point for the algorithm to make any degree queries.

\begin{lemma}
    Let $G_\yes \sim \yesdist$ and $G_\no \sim \nodist$. Then it holds with probability 1 that
    $$
    \mu(G_\yes) \geq (3-\epsilon) N \qquad \text{and} \qquad \mu(G_\no) \leq (2+2\epsilon) N.$$
    This, in particular, implies that any algorithm satisfying the condition of \cref{thm:lb} must be able to distinguish whether a graph $G$ drawn from distribution $(\frac{1}{2}\yesdist + \frac{1}{2} \nodist{})$ belongs to the support of \yesdist{} or \nodist{} with probability at least $.51$.
\end{lemma}
\begin{proof}
    For $G_{\yes}$, we take edges $(v_i(A_V), v_i(A_U))$ for $1\leq i \leq (1-\epsilon)N$, $(v_i(B_V), v_i(S_V))$ for $1\leq i \leq N$, and $(v_i(B_U), v_i(S_U))$ for $1\leq i \leq N$. Since none of the edges have the same endpoint, the union creates a matching with size $(3-\epsilon)N$ which implies $\mu(G_{\yes}) \geq (3 - \epsilon) N$.

    For $G_{\no}$, we take $B_U \cup B_V \cup T_U \cup T_V$ as a vertex cover. Note that there is no edge in the induced subgraph between vertices of $A_U \cup A_V \cup S_U \cup S_V$. The proof follows by K\"{o}nig’s Theorem since there exists a vertex cover with $(2 + 2\epsilon)N$ vertices.

    Furthermore, if $\alpha > 0$ be a constant, since $(2+2\epsilon) / (3-\epsilon) \leq 2/3 + \alpha$ (if we choose $\epsilon$ sufficiently smaller than $\alpha$), any algorithm satisfying the condition of \cref{thm:lb}, must be able to distinguish whether a graph $G$ drawn from distribution $(\frac{1}{2}\yesdist + \frac{1}{2} \nodist{})$ belongs to the support of \yesdist{} or \nodist{} with probability at least $.51$.
\end{proof}

\begin{remark}
    Note that our construction can be slightly modified to have a perfect matching in \yesdist{}. We can add a perfect matching between vertices of $T_U$ and $T_V$ to have a matching of size $3N$. Therefore, our lower bound also holds for the case that the \yesdist{} has a perfect matching.
\end{remark}

\subsection{Queried Edges Form a Rooted Forest}

Here and throughout the rest of the paper, we use $G'$ to denote the induced subgraph of $G$ excluding all the dummy vertices in $T_U \cup T_V$. Our main result of this section is that any algorithm that makes at most $\K{}$ queries to graph $G$, only discovers a rooted forest in $G'$:

\begin{lemma}\label{lem:forest}
    Let $\mc{A}$ be any algorithm making at most $K = \K{}$ queries to graph $G \sim \dist$. Let $F_0$ be the empty graph on the vertex set of $G'$, and for $t \geq 1$ let $F_t$ be the subgraph of $G'$ that $\mc{A}$ discovers after $t$ queries. The following property holds throughout the execution of $\mc{A}$ with probability $1-o(1)$: For any $t$, if the $t$'th query is to the adjacency list of a vertex $v$ and returns edge $(v, u)$ then either $(v, u) \not\in G'$ or if $(v, u) \in G'$ then $u$ is singleton in $F_{t-1}$. Equivalently, this implies that each $F_t$ can be thought of as a rooted forest where $F_t \setminus F_{t-1}$ may only include one edge $(v, u)$ and if $\mc{A}$ discovered $(v, u)$ by querying $v$, then $u$ becomes a leaf of $v$ in $F_t$.
\end{lemma}

\begin{remark}
Our proof of \cref{lem:forest} crucially relies on the fact that the adjacency list of each vertex is randomly permuted in the input distribution. However,  we emphasize here that our proof holds even if the internal permutation of the $G'$-neighbors of a vertex $v$ in the adjacency list of $v$ is adversarial. This implies, in particular, that if we condition on the high probability event of \cref{lem:forest}, the internal permutation of the $G'$-neighbors will still be uniform.
\end{remark}

To prove \cref{lem:forest}, we first bound the number of edges of $G'$ that we see with $K = \K{}$ queries.

\begin{claim}\label{clm:few-edges}
    Any algorithm $\mc{A}$ that makes $K = \K{}$ queries to $G$, discovers at most $o(n^{2/5})$ edges of $G'$ with probability $1-1/\poly(n)$.
\end{claim}
\begin{proof}
    Since we assume $\mc{A}$ makes at most $K = o(n^{6/5})$ queries, there will be at most $o(n^{1/5})$ vertices for which the algorithm makes more than $\epsilon N/2 = \Omega(n)$ adjacency list queries. For each of these $o(n^{1/5})$ vertices, we assume that we discover all their neighbors in $G'$ that is at most $d+1 = O(n^{1/5})$, which in total is $o(n^{2/5})$ edges.

   Now let $V'$ be the set of vertices for which $\mc{A}$ makes at most $\epsilon N /2$ queries.  For each new query to a vertex $v \in V'$, since there are $\epsilon N$ edges to $T_U \cup T_V$ in $G$ and that we have already made at most $\epsilon N/2$ queries to $v$, the new query goes to $G'$ with probability at most $\deg_{G'}(v)/(\deg_{G'}(v) + \epsilon N - \epsilon N/2) =  O(d/n)$. Next, let $X_i$ be the indicator of the event that the $i$'th query to $V'$ discovers a $G'$ edge. From our earlier discussion, we get $\E[X_i] = O(d/n)$. Moreover, these $X_i$'s are negatively correlated. Hence, denoting $X = \sum X_i$, we get $\E[X] \leq O(Kd/n)$ and can apply the Chernoff bound to obtain that with probability at least $1-1/\poly(n)$, 
   $$
   X \leq \E[X] + O(\sqrt{\E[X] \log n}) = O(Kd/n) = o(n^{2/5}),
   $$
   concluding the proof.
\end{proof}

Next, we show that conditioning on the fact that we can discover at most $o(n^{2/5})$ edges, the probability of having an edge between any pairs of vertices is at most $O(1/n^{4/5})$. To give an intuition of why this claim holds, assume that instead of the regular graphs between $B_U$ and $B_V$ (similarly between $A_U$ and $B_V$, and $A_V$ and $B_U$), we have an Erdos-Renyi with the same expected degree as the regular graphs. Since the degree of each regular graph is $O(d)$, then the probability of having an edge between a pair of vertices must be $O(d/n) = O(1/n^{4/5})$ since the existences of edges are independent. However, since we draw a random $O(d)$-regular graph, edge realizations are not independent the same argument does not work. But using careful coupling, we prove that the same claim holds for this construction.

\begin{claim}\label{clm:no-edge-between-edge}
    Let us condition on the high probability event of \cref{clm:few-edges} that the algorithm has discovered at most $o(n^{2/5})$ edges in $G'$. Then for any pair of vertices $u, v$ such that $(u, v)$ is not among the discovered edges, the probability that $(u, v)$ is an edge in $G'$ is at most $O(1/n^{4/5})$.
\end{claim}

\begin{proof}
    Let $A = A_U \cup A_V$ and $B = B_U \cup B_V$. Also, assume that by index of a vertex we mean the index in the permutation of its corresponding subset in the construction. There are three possible cases for the type of subsets that $u$ and $v$ belong to:

   \paragraph{(Case 1) $u, v \in A$:} note that since the algorithm makes at most $o(n^{6/5})$ queries, using Chernoff bound, it is easy to see that at least a constant fraction of perfect matching edges in the induced subgraph of $A$ remains undiscovered. If $u$ and $v$ both belong to $A$, the probability of having an edge between the pair is at most $O(1/n)$ since we only have a perfect matching between vertices with the same indices in $A_U$ and $A_V$ for $\mc{D}_{\yes}$, and vertices are randomly permuted.

    For the other two cases, we bound the probability of having an edge between $u$ and $v$ by considering two possible scenarios. First, the probability that both $u$ and $v$ have the same indices is equal to $1/n$. Next, assume that $u$ and $v$ have different indices. Let $\mc{G}_e$ be the set of all graphs in $\mc{D}$ such that all of them have the edge $e = (u,v)$. Also, let $\overline{\mc{G}}_e$ be the set of all graphs in $\mc{D}$ with no edge between pair $(u,v)$. In the next two cases, we use coupling to show that conditioned on discovered edges, we have $|\mc{G}_e| / |\overline{\mc{G}}_e| \leq O(1/n^{4/5})$ which implies that the probability of having an edge for a pair $(u, v)$ is at most $O(1/n^{4/5})$.

   \paragraph{(Case 2) $u, v \in B$:} without loss of generality, assume that $u \in B_U$ and $v \in B_V$. First, note that with a probability of $O(1/n)$, the index of $u$ and $v$ is the same, which implies that there is an edge between $u$ and $v$. Let $i_x$ denote the index of vertex $x$. Now assume that $i_u \neq i_v$. We claim if there exists edge $(u,v)$, then there are at least $N\epsilon d /2$ edges $(w,z)$ such that $z \in B_U$, $w \in B_V$, $\{i_u, i_v, i_w, i_z \} = 4$, and the induced subgraph of these four vertices only has two edges $(u,v)$ and $(w, z)$.

   Since the number of incident edges of $u$ to $B_V$ is at most $\epsilon d$, then there are at least $N - \epsilon d$ vertices in $B_V$ that are not connected to $u$. Also, we exclude the vertex with index $i_u$ in $B_V$ from this set. Let $Q$ be the set of such vertices in $B_V$. Hence, $|Q| \geq N -\epsilon d - 1$. Furthermore, by the construction, each vertex $w \in Q$ has at least $\epsilon d - 3$ neighbors in $B_U$ with an index not in $\{i_u, i_v, i_w\}$. Hence, there are at least $(\epsilon d - 3)(N -\epsilon d - 1)$ edges between vertices of $Q$ and $B_U$ with endpoints having different labels. Moreover, at most $\epsilon^2 d^2$ of these edges are incident to an edge that one of its endpoints is $v$. Therefore, there are at least $(\epsilon d - 3)(N -\epsilon d - 1) - \epsilon^2 d^2 \geq N\epsilon d /2$ edges $(w,z)$ that satisfy the claimed properties where the inequality follows by the choice of $d$.

   Now assume that we remove all edges that we discovered so far. Since we discover at most $o(n^{2/5})$ edges, by removing these edges we still have $O(n\epsilon d)$ edges $(w,z)$ with mentioned properties. Let $(w,z)$ be such an edge where $w \in B_V$ and $z \in B_U$. We construct a graph by removing edge $(u, v)$ and $(z, w)$, and adding edges $(u, w)$ and $(v, z)$. If the original graph is in $\mc{D}_{\yes}$ ($\mc{D}_{\no}$), the new graph is in $\mc{D}_{\yes}$ ($\mc{D}_{\no}$) according to the construction since we only change the random regular bipartite graph between $B_U$ and $B_V$ without changing the degrees. We construct a bipartite graph where each vertex of the first part represents a graph in $\mc{G}_e$ and each vertex in the second part represents a graph in $\overline{\mc{G}}_e$. For each vertex in the first part, we connect it to the vertices of the second part that can be produced by the above operation. Thus, the degree of vertices in the first part is at least $O(n\epsilon d)$. On the other hand, each vertex in the second part is connected to at most $O(\epsilon^2 d^2)$ vertices of the first part since the degree of each $u$ and $v$ is at most $\epsilon d$ in the induced subgraph of $B$. Counting the edges from both sides of the constructed bipartite graph yields
   \begin{align*}
       \frac{|\mc{G}_e|}{|\overline{\mc{G}}_e|} \leq \frac{O( \epsilon^2 d^2)}{O(N\epsilon d)} = O\left(\frac{\epsilon d}{N}\right) = O_\epsilon\left(\frac{1}{n^{4/5}}\right).
   \end{align*}

    \paragraph{(Case 3) $u \in A, v \in B$ or $u \in B, v \in A$:} without loss of generality, assume that $u \in A_U$ and $v \in B_V$. With the same argument as the previous case, the probability that both $u$ and $v$ have the same index is $O(1/n)$. Moreover, since the degree of the induced subgraph of $A_U \cup B_V$ differs from the previous case by a constant factor, with the same reasoning as the previous part, the probability that there exists an edge between $u$ and $v$ is at most $O(1/n^{4/5})$.

    Therefore, conditioning on the discovered edges, the probability of having an edge $(u,v)$, is at most $O(1/n^{4/5})$
    \end{proof}

Now we are ready to complete the proof of \Cref{lem:forest}.

\begin{proof}[Proof of \Cref{lem:forest}]
    First, we prove that for any $t \geq 1$, if two vertices $u, v$ are non-singleton in $F_t$ and $(u, v) \not\in F_t$, then there is no edge between $u$ and $v$ in graph $G'$. This will imply that any time that we discover a new edge of a non-singleton vertex in $F_t$, it must go to a vertex that is singleton in $F_t$. We prove this by induction on $t$. For $F_0$, the graph is empty and so the claim holds. Suppose that there are no undiscovered edges among the non-singleton vertices of $F_{t-1}$. We prove that this continues to hold for $F_t$. If $F_t \setminus F_{t-1} = \emptyset$, i.e., if we do not discover any new edge of $G'$ at step $t$, then the claim clearly holds. So suppose that we query some vertex $v$ and find an edge $(v, u)$ at step $t$. It suffices to show that in this case, $u$ will not have any non-singleton neighbors in $F_{t-1}$ other than $v$. To see this, recall by \Cref{clm:few-edges} that there are at most $o(n^{2/5})$ non-singleton vertices in $F_{t-1}$. Fixing any such vertex $w$ with $w \not= v$, we get by \cref{clm:no-edge-between-edge} that the conditional probability of $(u, w)$ being an edge in $G'$ is at most $O(1/n^{4/5})$. By a union bound over all $o(n^{2/5})$ choices of $w$, we get that the probability of $u$ having an edge to any of the non-singleton vertices of $F_{t-1}$ is $o(n^{2/5}) \cdot O(1/n^{4/5}) = o(1/n^{2/5})$. Finally, since by \cref{clm:few-edges} there are at most $O(n^{2/5})$ steps where we discover any edge of $G'$, the failure probability over all the steps of the induction is $O(n^{2/5}) \cdot o(1/n^{2/5}) = o(1)$. Hence, the claim is true throughout with probability at least $1-o(1)$.
    
    We proved above that by querying an already non-singleton vertex $v$, its discovered neighbor $u$ must be singleton w.h.p. We will now prove that this also holds if $v$ itself is singleton. As discussed, the number of non-singleton vertices in $F_t$ is  $o(n^{2/5})$ by \cref{clm:few-edges}. For each of these vertices $u$, the conditional probability of $v$ having an edge to $u$ is at most $O(1/n^{4/5})$ by \cref{clm:no-edge-between-edge}. Hence, the expected number of edges of $v$ to non-singleton vertices is $o(1/n^{2/5})$. Since $v$ has $\Omega(n)$ edges in $G$ and its adjacency list is uniformly sorted, the probability of discovering one of such edges of $v$ is at most $o(1/n^{2/5}) / n = o(1/n^{7/5})$. A union bound over all $o(n^{6/5})$ queries of the algorithm implies that with probability $1-o(1)$, any time that we query a singleton vertices of $F_t$, its discovered edge will not be to another non-singleton vertex of $F_t$.
\end{proof}

\subsection{The Tree Model}

In our proofs, we condition on the high probability event of \cref{lem:forest} that $F_t$ for any $t$ forms a rooted forest. Under this event, we prove that  the algorithm cannot distinguish the \yes{} distribution from the \no{} distribution.

Recall that from our definition in \cref{lem:forest}, every vertex in $F_t$ for any $t$ is a vertex in set $A \cup B \cup S$. Given the forest $F_t$, we do not necessarily know for a given a vertex $v$ in $F_t$ whether it belongs to $A$ or $B$ (but if it belongs to $S$ we know this since recall all the $S$ vertices are given for free by the distribution). Inspired by this, we can see each vertex $v$ in $F_t$ as a random variable taking one of the values $\{A, B, S\}$. Importantly, depending on the value of $v$ we have a different distribution on what values its children in the tree will take. This distribution is also different for the \yes{} and \no{} distributions. For example, in the \no{} distribution all children of any vertex $v \in A$ will be $B$ vertices whereas in the \yes{} distribution there is a small chance of seeing an $A$ child. We prove that although these distributions are different for the \yes{} and \no{} distributions, no algorithm can distinguish (with a sufficiently large probability) whether the observed forest $F_t$ was drawn from \yesdist{} or \nodist{}.

The proof consists of multiple steps. First, in Section~\cref{sec:cor-decay}, we prove a correlation decay property on the tree. Then equipped with the correlation decay property, we are able to show the probability of seeing the same forest by the algorithm in a \yesdist{} and \nodist{} is almost the same which implies that the algorithm cannot decide if the graph is drawn from \yesdist{} or \nodist{}.

\subsection{Correlation Decay}\label{sec:cor-decay}

\cref{lem:cor-decay} below is our main result in this section. Before stating the lemma, let us start with some definitions. Consider a tree $T(v)$ of any arbitrary depth rooted at vertex $v$. We say $T(v)$ is $\ell$-$S$-free if no vertex in $T(v)$ with distance at most $\ell$ from the root belongs to $S$. We also use $Q(T(v), \ell)$ to denote the total number of vertices of distance at most $\ell$ from the root $v$ in $T(v)$. With a slight abuse of notation, we may also use $T(v)$ to denote the event that the sub-tree of $F_t$ rooted at vertex $v$ by the end of the algorithm (i.e., when $t = K$) is exactly the same as $T(v)$.

\begin{lemma}\label{lem:cor-decay}
    Let $v$ be any vertex, $\ell \geq 20(\log^2 n)/\epsilon$, and let $T(v)$ be any $\ell$-$S$-free tree. Let $P(v)$ and $P'(v)$ each be an arbitrary outcome of $v$ from $\{A, B\}$ and the entire forest $F_t(v)$ excluding the sub-tree of $v$. Letting $Q := Q(T(v), \ell)$ and assuming that $Q = o(d)$, we have
    $$
        \Pr_{\yesdist{}}[T(u_{\ell'}) \mid P(v)] \leq  \Big(1+\frac{1}{d}\Big)^{O(Q)}\Pr_{\yesdist{}}[T(u_{\ell'}) \mid P'(v)],
    $$
    and
    $$
        \Pr_{\nodist{}}[T(u_{\ell'}) \mid P(v)] \leq  \Big(1+\frac{1}{d}\Big)^{O(Q)} \Pr_{\nodist{}}[T(u_{\ell'}) \mid P'(v)].
    $$
\end{lemma}

\begin{proof}
    Call a vertex $u \in T(v)$ an {\em internal} vertex if it has distance at most $\ell$ from $v$. Note that $Q$ is the number of internal vertices in $T(v)$. We say a subset of internal vertices $u_1, \ldots, u_k$ form an {\em internal path} in a tree $T$ if for any $1 \leq i \leq k-1$, $u_{i}$ is the parent of $u_{i+1}$ in $T$, and additionally each $u_i$, for $i \in [k]$, only has one child in $T$. 
    
    We use the following two auxiliary \cref{cl:internal-paths,cl:cor-decay-on-paths} to prove \cref{lem:cor-decay}.
    
    \begin{claim}\label{cl:internal-paths}
        Suppose that every vertex in $T(v)$ has at least one non-internal descendant. Then $T$ must have an internal path of length at least $10\log n/\epsilon$.
    \end{claim}
    \begin{myproof}
        Suppose for contradiction that $T(v)$ does not have any internal path of length $\ell'$. Construct a tree $T'(v)$ by contracting all maximal internal paths of $T(v)$ (of any length). Since every vertex in $T(v)$ has a non-internal descendant, which is of distance at least $\ell$ from the root by definition, and every contracted path has a length less than $\ell'$, every vertex in $T'(v)$ must have a descendant of distance at least $\ell/\ell'$ from the root. Furthermore, every vertex in $T'(v)$ must have at least two children. As such, we get that $T'(v)$ must have $2^{\ell/\ell'} > n$ vertices, a contradiction.
    \end{myproof}
    
    The proof of \cref{cl:cor-decay-on-paths} below is involved; we state it here and present the proof in \cref{sec:cor-decay-paths}.
    
    \begin{claim}\label{cl:cor-decay-on-paths}
        Let $(u_1, \ldots, u_{\ell'})$ be an internal path in $T(v)$ of length $\ell' \geq 10\log n / \epsilon$ such that each $u_i$ is the parent of $u_{i+1}$. Let $T(u)$ be the sub-tree of $T(v)$ rooted at vertex $u$ and let $P(u_1)$ and $P'(u_1)$ be any two events on the vertices outside $u_1$'s sub-tree. Then 
        $$
            \Pr_{\yesdist{}}[T(u_{\ell'}) \mid P(u_1)] \leq  \Big(1+\frac{1}{d}\Big)^{O(\ell')}\Pr_{\yesdist{}}[T(u_{\ell'}) \mid P'(u_1)],
        $$
        and
        $$
            \Pr_{\nodist{}}[T(u_{\ell'}) \mid P(u_1)] \leq  \Big(1+\frac{1}{d}\Big)^{O(\ell')} \Pr_{\nodist{}}[T(u_{\ell'}) \mid P'(u_1)].
        $$
    \end{claim}
    
    We are now ready to present the proof of \cref{lem:cor-decay}. Given the tree $T(v)$, to measure the probability of $T(v)$ actually happening we start by querying from $v$ the same tree topology. That is, if some vertex $u$ has $x$ children in $T(v)$, we reveal $x$ children of $u$ from the distribution and measure the probability that the resulting tree is exactly the same as $T(v)$. First, let us measure the probability that the resulting tree is indeed $\ell$-$S$-free. That is, no internal vertex in the queried subtree is an $S$ vertex. To do this, take an internal vertex $u$ in $T(v)$. Conditioned on either of $P(u_1)$ or $P'(u_1)$, the probability of $u$ being an $S$ vertex is at most $O(1/d)$ since every vertex in $G'$ has at most one $S$ neighbor, has degree $d+1$, and its adjacency list is uniformly sorted. On the other hand, since the total number of internal vertices in $T(v)$ is $Q$, by a union bound, the probability of seeing at least one $S$ vertex is at most $O(Q/d)$. Let us condition on the event that we see no internal $S$ vertex. This only multiplies the final probability of $T(v)$ occuring by some $1 \pm O(Q/d) \leq (1 \pm 1/d)^{O(Q)}$ factor.

    Conditioned on the observed tree being $\ell$-$S$-free, note that if some internal vertex $u$ in the tree has no non-internal descendants, then indeed its sub-tree exactly matches $T(v)$. So let us peel off all such sub-trees, arriving at a tree where every internal vertex has at least one non-internal descedntant. From \cref{cl:internal-paths}, we get that there must be a an internal path $u_1, \ldots, u_{\ell'}$ of length at least $\ell' \geq 10\log n / \epsilon$ in this tree. Applying \cref{cl:cor-decay-on-paths}, the probability of $T(u_{\ell'})$ happening remains the same up to a factor of $(1 + 1/d)^{O(\ell')}$, no matter what we condition on above $u_1$. Thus, we can peel off the whole sub-tree of $u_{1}$ and repeat. Assuming that we repeat this process $K$ times until we peel off all vertices, we get that the probability of $T(v)$ occuring under $P(v)$ and $P'(v)$ is the same up to a factor of $(1 \pm 1/d)^{O(Q)} \times (1\pm 1/d)^{\ell' K} = (1 \pm 1/d)^{O(Q + \ell'K)}$. Finally, note that $\ell' K = O(Q)$ since every time we peel off the sub-tree of an internal path, we remove $\ell'$ internal vertices in the path, and so $K \times \ell'$ is at most the number of internal vertices. This completes the proof of \cref{lem:cor-decay}.
    \end{proof}
    
    \subsubsection{Correlation Decay on Internal Paths: Proof of \cref{cl:cor-decay-on-paths}}\label{sec:cor-decay-paths}
    
    In this section, we present the proof of \cref{cl:cor-decay-on-paths}, which as discussed implies correctness of \cref{lem:cor-decay}.
    
    \begin{proof}[Proof of \cref{cl:cor-decay-on-paths}]
        For any $i \in [\ell']$ we define
        $$
            b_i = \Pr[u_i \in B \mid P(u)], \qquad a_i = \Pr[u_i \in A \mid P(u)].
        $$
        We claim that for both distributions \yesdist{} and \nodist{}, and for any $i > 1$, 
        \begin{flalign}\label{eq:recursive}
            b_i \in \left(1 \pm \frac{2}{\epsilon d}\right) \left( \epsilon b_{i-1} + a_{i-1} \right), \qquad a_i \in \left(1 \pm \frac{2}{\epsilon d}\right) (1-\epsilon) b_{i-1} \pm  \frac{a_{i-1}}{d}.
        \end{flalign}
        To see this, let us first focus on $b_i$. We have 
        \begin{flalign}
            \nonumber \Pr[u_i \in B \mid P(u)] &= \Pr[u_i \in B \mid u_{i-1} \in B, P(u)] \cdot \Pr[u_{i-1} \in B \mid P(u)]\\
            \nonumber &+ \Pr[u_i \in B \mid u_{i-1} \in A, P(u)] \cdot \Pr[u_{i-1} \in A \mid P(u)]\\
            &= \Pr[u_i \in B \mid u_{i-1} \in B, P(u)] b_{i-1} + \Pr[u_i \in B \mid u_{i-1} \in A, P(u)] a_{i-1}.\label{eq:hcjj-299283780}
        \end{flalign}
        Let us now examine $\Pr[u_{i} \in B \mid u_{i-1} \in B, P(u)]$. Given that $u_{i-1} \in B$, vertex $u_{i-1}$ has exactly $d+1$ neighbors in graph $G'$. Among them, either $\epsilon d$ or $\epsilon d - 1$ neighbors of $u_i$ are in $B$. Additionally, since $u_{i-1}$ has exactly one child (by definition of internal paths) then the conditional event $P(u)$ may only reveal the $A$/$B$-value of one neighbor of $u_{i-1}$, which would be its parent. Since the adjacency list of vertex $u_{i-1}$ is randomly sorted, its discovered child $u_{i}$ belongs to $B$ with probability at least $\frac{\epsilon d - 2}{d} \geq (1 - \frac{2}{\epsilon d}) \epsilon$ and at most $\frac{\epsilon d}{d} = \epsilon$. Thus, 
        $$\Pr[u_{i} \in B \mid u_{i-1} \in B, P(u)] \in (1 \pm \frac{2}{\epsilon d}) \epsilon.$$
        On the other hand, since every vertex in $A$ has at least $d$ and at most $d+1$ neighbors in $A$,
        $$\Pr[u_{i} \in B \mid u_{i-1} \in A, P(u)] \in (1 \pm \frac{1}{d}).$$
        Replacing these two bounds in \cref{eq:hcjj-299283780}, we indeed arrive at the recursion of \cref{eq:recursive} for $b_i$.
        
        The calculation for $a_i$ is similar. In particular, as in \cref{eq:hcjj-299283780}, we have
        \begin{flalign}
            \nonumber \Pr[u_i \in A \mid P(u)] &= \Pr[u_{i} \in A \mid u_{i-1} \in B, P(u)] b_{i-1} +  \Pr[u_{i} \in A \mid u_{i-1} \in A, P(u)] a_{i-1}.
        \end{flalign}
        We have $\Pr[u_{i} \in A \mid u_{i-1} \in B, P(u)] \in (1 \pm \frac{2}{\epsilon d}) (1-\epsilon)$ since almost $(1-\epsilon)$ fraction of neighbors of each $B$ vertex go to $A$, and we have $\Pr[u_{i} \in A \mid u_{i-1} \in A, P(u)] \leq 1/d$ since each $A$ vertex has at most one $A$ neighbor among its $d+1$ neighbors. Replacing these into the inequality above, implies the recursion of \cref{eq:recursive} for $a_i$.
        
        The following claim gives an explicit (i.e., non-recursive) bound for $b_i$ and $a_i$ as a function of just $a_1$ and $b_1$.
    
        \begin{claim}\label{cl:explicit}
        Let 
        $$
            a'_i := \frac{(1-\epsilon)^2(a_1+b_1) - (\epsilon - 1)^i (a_1 - (1-\epsilon)b_1) }{(2-\epsilon)(1-\epsilon)},
        $$ 
        and
        $$
            b'_i := \frac{(1-  \epsilon)(a_1 + b_1) + (\epsilon - 1)^i (a_1 - (1-\epsilon)b_1)}{(2-\epsilon)(1-\epsilon)}.
        $$
        Then for any $i \geq 3$, we have
        \begin{flalign*}
            b_i \in \left( 1 \pm \frac{2}{\epsilon d} \right)^{2i} b'_i, \qquad 
            a_i \in \left( 1 \pm \frac{2}{\epsilon d} \right)^{2i} a'_i.
        \end{flalign*}
        \end{claim}
        \begin{proof}
            First, we note a couple of useful properties of $b'_i$ and $a'_i$ that all can be verified from their definitions:
            \begin{equation}\label{eq:bpap}
                a'_1 = a_1, \qquad b'_1 = b_1, \qquad b'_i = \epsilon b'_{i-1} + a'_{i-1}, \qquad a'_i = (1-\epsilon)b'_{i-1}.
            \end{equation}
            To prove the stated bounds on $b_i$ and $a_i$ we use induction on $i$. For the base case $i=3$, directly applying \cref{eq:recursive} we get
            \begin{flalign*}
                b_3 &\in \left(1 \pm \frac{2}{\epsilon d} \right) (\epsilon b_{2} + a_{2})\\
                &\in \left(1 \pm \frac{2}{\epsilon d} \right) \left[ \epsilon\left( \left(1 \pm \frac{2}{\epsilon d} \right) (\epsilon b_{1} + a_{1}) \right) + \left(\left(1 \pm \frac{2}{\epsilon d} \right)(1-\epsilon)b_1 \pm \frac{a_1}{d} \right) \right]\\
                &\in \left(1 \pm \frac{2}{\epsilon d} \right)^{2} \left[ \epsilon^2 b_1 + \epsilon a_1 + (1-\epsilon) b_1 \pm \frac{a_1}{d}\right]\\
                &\in \left(1 \pm \frac{2}{\epsilon d} \right)^{3} \left[ (1 - \epsilon + \epsilon^2) b_1 + \epsilon a_1 \right]\\
                &= \left(1 \pm \frac{2}{\epsilon d} \right)^{3} b'_3 \in \left(1 \pm \frac{2}{\epsilon d} \right)^{6} b'_3.
            \end{flalign*}
            Similarly, we have
            \begin{flalign*}
                a_3 &\in \left( 1 \pm \frac{2}{\epsilon d} \right) (1-\epsilon) b_2 \pm \frac{a_2}{d}\\
                &\in \left( 1 \pm \frac{2}{\epsilon d} \right) (1-\epsilon) \left( \left(1 \pm \frac{2}{\epsilon d} \right) (\epsilon b_{1} + a_{1}) \right) \pm \frac{1}{d} \left( \left( 1 \pm \frac{2}{\epsilon d} \right) (1-\epsilon) b_1 \pm \frac{a_1}{d} \right)\\
                &\in \left( 1 \pm \frac{2}{\epsilon d} \right)^3 (1-\epsilon) \left( \epsilon b_1 + a_1 \right)\\
                &= \left( 1 \pm \frac{2}{\epsilon d} \right)^3 a'_3 \in \left( 1 \pm \frac{2}{\epsilon d} \right)^6 a'_3.
            \end{flalign*}
            
            We now turn to prove the induction step for $i$, assuming that it holds for $i-1$. Let us start with $b_i$. We have
            \begin{flalign*}
                b_i &\in \left(1 \pm \frac{2}{\epsilon d} \right) (\epsilon b_{i-1} + a_{i-1}) \tag{From \cref{eq:recursive}.}\\
                &\in \left(1 \pm \frac{2}{\epsilon d} \right) \left( \epsilon \left(1 \pm \frac{2}{\epsilon d} \right)^{2(i-1)} b'_{i-1} + \left(1 \pm \frac{2}{\epsilon d} \right)^{2(i-1)} a'_{i-1} \right) \tag{By the induction hypothesis.}\\
                &\in \left(1 \pm \frac{2}{\epsilon d} \right)^{2i} \left( \epsilon b'_{i-1} + a'_{i-1} \right)\\
                &= \left(1 \pm \frac{2}{\epsilon d} \right)^{2i} b'_i. \tag{By \cref{eq:bpap}.}
            \end{flalign*}
            Moreover, for $a_i$ we have
            \begin{flalign*}
                a_i &\in \left(1 \pm \frac{2}{\epsilon d}\right) (1-\epsilon) b_{i-1} \pm \frac{a_{i-1}}{d} \tag{From \cref{eq:recursive}.}\\
                &\in \left(1 \pm \frac{2}{\epsilon d}\right)^{2(i-1)+1} (1-\epsilon) b'_{i-1} \pm \left(1 \pm \frac{2}{\epsilon d}\right)^{2(i-1)} \frac{a'_{i-1}}{d} \tag{By the induction hypothesis.}\\
                &\in \left(1 \pm \frac{2}{\epsilon d}\right)^{2i} (1-\epsilon) b'_{i-1} \tag{Since $a'_{i-1} \leq b'_{i-1}/\epsilon$ for $i \geq 4$.}\\
                &\in \left(1 \pm \frac{2}{\epsilon d}\right)^{2i} a'_i. \tag{By \cref{eq:bpap}.}
            \end{flalign*}
            This completes the proof of \cref{cl:explicit}.
        \end{proof}
        
        Observe that $a_1 + b_1 = 1$ since $u_1 \not\in S$ and so either $u_1 \in A$ or $u_1 \in B$. Combined with \cref{cl:explicit}, we get that 
        \begin{flalign*}
            b_{\ell'} &\in \left( 1 \pm \frac{2}{\epsilon d} \right)^{2\ell'} \frac{(1-  \epsilon)(a_1 + b_1) + (\epsilon - 1)^{\ell'} (a_1 - (1-\epsilon)b_1)}{(2-\epsilon)(1-\epsilon)}\\
            &\in \left( 1 \pm \frac{2}{\epsilon d} \right)^{2\ell'} \frac{1-\epsilon + (\epsilon - 1)^{10\log n /\epsilon}(a_1 - (1-\epsilon)b_1)}{(2-\epsilon)(1-\epsilon)} \tag{Since $\ell' \geq 10\log n \epsilon$.}\\
            &\in \left( 1 \pm \frac{2}{\epsilon d} \right)^{2\ell'+1} \frac{1-\epsilon}{(2-\epsilon)(1-\epsilon)}\\
            &= \left( 1 \pm \frac{1}{d} \right)^{O(\ell')} \frac{1}{2-\epsilon}.
        \end{flalign*}
        Observe that the conditioning $P(u_1)$ determines the values of $a_1$ and $b_1$. However, for large enough $\ell'$, as we see above, the dependence of $b_{\ell'}$ on $a_1$ and $b_1$ vanishes. In other words, by changing the conditioning $P(u_1)$ the value of $b_{\ell'}$ for $\ell' \geq 10\log n / \epsilon$ only changes by a $(1 \pm 1/d)^{O(\ell')}$ factor. The same also holds for $a_{\ell'}$. As such, whether vertex $u_{\ell'}$ belongs to $A$ or $B$ is essentially independent of the events above the root vertex $u_1$. To see why this implies \cref{cl:cor-decay-on-paths}, take for example the \yesdist{} distribution and note that
        \begin{flalign*}
            \Pr_{\yesdist{}}[T(u_{\ell'}) \mid P(u_1)] &= \Pr_{\yesdist{}}[u_{\ell'} \in A \mid P(u_1)] \cdot \Pr_{\yesdist{}}[T(u_\ell') \mid u_{\ell'} \in A]\\
            &+ \Pr_{\yesdist{}}[u_{\ell'} \in B \mid P(u_1)] \cdot \Pr_{\yesdist{}}[T(u_\ell') \mid u_{\ell'} \in B]\\
            &= a_{\ell'} \Pr_{\yesdist{}}[T(u_\ell') \mid u_{\ell'} \in A] + b_{\ell'} \Pr_{\yesdist{}}[T(u_\ell') \mid u_{\ell'} \in B].
        \end{flalign*}
        Since $a_{\ell'}$ and $b_{\ell'}$ remain the same up to a $(1 \pm 1/d)^{O(\ell')}$ factor by changing the conditioning $P(u_1)$ to $P'(u_1)$, we arrive at the desired inequality of \cref{cl:cor-decay-on-paths} for the \yesdist{} distribution. The proof for \nodist{} is exactly the same.
    \end{proof}

\subsection{Limitation of the Algorithm}

Let us define $E_A$ to be the set of all discovered edges of $G[A]$ by the algorithm. Also, let $V_A = \{ v \mid (u, v) \in E_A\}$, where here $(u, v)$ is directed and $u$ is the parent of $v$ in the forest of discovered edges.

\begin{claim}\label{clm:few-bs-edges}
    With high probability, any algorithm $\mc{A}$ that makes at most $o(n^{6/5} / \log^2 n)$ queries, discovers at most $o(d / \log^2 n)$ of edges in subgraph $G'$ with one endpoint in $S$.
\end{claim}
\begin{proof}
    The proof is similar to the proof of \Cref{clm:few-edges}. Since we assume $\mc{A}$ makes at most $o(n^{6/5} / \log^2 n)$ queries, there will be at most $o(n^{1/5} /\log^2 n)$ vertices for which the algorithm makes more than $\epsilon N/2 = \Omega(n)$ adjacency list queries. For these vertices, we assume that $\mc{A}$ finds its edge in $G'$ with one endpoint in $S$ (note that the degree of $S$ vertices in $G'$ is one and each vertex of $G'$ is connected to at most one vertex of $S$).
    
    Now let $V'$ be the set of vertices for which $\mc{A}$ makes at most $\epsilon N /2$ queries.  For each new query to a vertex $v \in V'$, since there are $\epsilon N$ edges to $T_U \cup T_V$ in $G$ and that we have already made at most $\epsilon N/2$ queries to $v$, the new query goes to a vertex in $S$ with probability at most $O(1/n)$. Hence, by applying the Chernoff bound, with probability at least $1-1/\poly(n)$, there are at most $o(d / \log^2 n)$ edges in subgraph $G'$ with one endpoint in $S$.
\end{proof}

\begin{lemma}\label{lem:always-few-in-subtree}
    Suppose that the algorithm $\mc{A}$ has made $t$ queries for some $t \leq O(n^{6/5} / \log^2)$. Then with probability of $1 - o(1)$, all the following hold:
    \begin{enumerate}
        \item[(i)] $\mc{A}$ has found at most $o(d)$ vertices in total for all subtrees $T(v)$ for $v \in V_A$ up to a distance $20(\log^2 n)/\epsilon$ from root $v$.
        \item[(ii)] $\mc{A}$ has not found any edge in $G'$ with one endpoint in $S$ in subtree $T(v)$ up to distance $20(\log^2 n)/\epsilon$ from root $v$, for all $v \in V_A$.
        \item[(iii)] Conditioning on what $\mc{A}$ has queried so far, the probability of each edge in $F_t$ belonging to $G[A]$ is $O(1/d)$.
    \end{enumerate}
\end{lemma}

\begin{proof}
    We prove the lemma using induction. For $t = 0$, trivially the statement is true. Now assume that the lemma holds for $t - 1$ and $\mc{A}$ makes a new query. If the new queried edge is an edge to $T_U \cup T_V$, we are done since none of the conditions in the lemma statement will change. So we assume that the newly queried edge is in $G'$. We prove each of the three claims separately.

    \paragraph{Induction step for (ii):} if the newly queried edge is between $B$ and $S$, the probability that one of its $20\log^2 n /\epsilon$ most recent predecessors  
    is an edge in $G[A]$ is $O(\log^2 n)/ d)$ using the induction hypothesis (iii). Since $\mc{A}$ can discover at most $o(d / \log^2 n)$ edges in $G'$ with one endpoint in $S$ by \Cref{clm:few-bs-edges}, then with probability at least $1- o(\frac{\log^2 n}{d} \cdot \frac{d}{\log^2 n}) = 1 - o(1)$, the statement of (ii) remains true throughout all the steps of the induction.

    \paragraph{Induction step for (i):} the probability that at least one of the $20\log^2 n /\epsilon$ most recent predecessors of the newly queried edge is an edge in $G[A]$ is $O(\log^2 n / d)$ using the induction hypothesis (iii). Since $\mc{A}$ discovers at most $o(n^{2/5} / \log^2 n) = o(d^2/\log^2 n)$ edges of $G'$ by \Cref{clm:few-edges}, then with probability at least $1 - o(\frac{\log^2 n }{d^2} \cdot \frac{d^2}{\log^2 n})= 1 - o(1)$, the statement of (i) remains true throughout all the steps of the induction.

    \paragraph{Induction step for (iii):} let $(u, v)$ be a discovered edge in the forest ($u$ is parent of $v$). First, we condition on $u \in A$, otherwise, the probability of $(u,v)$ being an edge in $G[A]$ is zero. Note that $u$ has at least $d$ neighbors in $G'$ that are not the parent of $u$ in the forest $F_t$. By \Cref{clm:few-bs-edges}, \mc{A} discovers at most $o(d)$ edges in subgraph $G'$ with one endpoint in $S$. Hence, at most $o(d)$ of neighbors of $u$ in $G'$ have an edge with one endpoint in $S$ in their subtree. Furthermore, by \Cref{clm:few-edges}, at least $\Theta(d)$ neighbors of $u$ in $G'$, have at most $o(d)$ vertices in their subtree. Also, note that by proof of \Cref{lem:forest}, neighbors of $u$ in $G'$ that are not adjacent to $u$ in the forest are singleton vertices in the forest. Let $v_1, v_2, \ldots, v_r$ be the union of (1): children of $u$ in the forest such that each of them has no edge with one $S$ endpoint in their subtree up to distance $20 \log^2 n /\epsilon$, and each of them has $o(d)$ vertices in their subtree up to distance $20 \log^2 n /\epsilon$, (2): the neighbors of $u$ in $G'$ that are singleton in the forest. Hence, we have $r = \Theta(d)$.

    By (i) and (ii), if $v \in A$ then there are at most $o(d)$ vertices in subtree $T(v)$ up to distance $20 \log^2 n /\epsilon$ from root $v$, and there is no edge in subtree $T(v)$ with one endpoint in $S$ up to distance $20 \log^2 n /\epsilon$ from the root. Thus, if $v \notin \{v_1, v_2, \ldots v_r\}$, then $(u,v)$ is not an edge in $G[A]$. Now assume that $v \in \{v_1, v_2, \ldots v_r\}$. Note that in \nodist{}, there is no neighbor of $u$ with label $A$ and in \yesdist{}, there is exactly one neighbor in $A$. Therefore, for \nodist{}, the probability of $v \in A$ is zero. Now assume that the graph is drawn from \yesdist{}. Hence, exactly one of $v_1, v_2, \ldots, v_r$ is in $A$ (if parent of $u$ is $A$ then the probability of $(u, v)$ being in $G[A]$ is zero). We prove that each of $v_i$ can be in $A$ with almost the same probability using a coupling argument.

    Let $c_i$ be the subset that vertex $v_i$ belongs to ($c_i \in \{A, B\})$. We say $C = (c_1, c_2, \ldots, c_r)$ is a {\em profile} for labels of vertices $v_1, \ldots v_r$. Therefore, exactly one of $c_i$ is equal to $A$ and all others are equal to $B$ which implies that there are $r$ different possible profiles. Assume that $C$ and $C'$ are two different profiles. Let $v_i$ be a vertex with label $A$ in $C$ and $v_j$ be a vertex with label $A$ in $C'$. By \Cref{lem:cor-decay}, the probability of sampling subtree below $v_i$ with label $B$ is the same up to a factor of $\left(1 \pm \frac{1}{d}\right)^{O(Q(T(v_i), \ell))}$. Similarly, the probability of sampling subtree below $v_j$ with label $A$ is the same up to a factor of $\left(1 \pm \frac{1}{d}\right)^{O(Q(T(v_j), \ell))}$. Since $Q(T(v_j), \ell) = o(d)$, then the probability of having profile $C$ and $C'$ is the same up to a $(1 + o(1))$ factor. Therefore, the probability of having one specific $v_i$ to be in $A$ is $O(1/d)$ since $r = \Theta(d)$.
\end{proof}

\subsection{Indistinguishability of the \yes{} and \no{} distributions}

We define a \emph{bad event} to be the event of $\mc{A}$ discovering more than $o(d)$ vertices in total for all subtrees $T(v)$ for $v \in V_A$ up to a distance $20 \log^2 n /\epsilon$ from root $t$, or $\mc{A}$ has found at least one edge in $G'$ with one endpoint in $S$ in subtree $T(v)$ up to distance $20 \log^2 n /\epsilon$ from root $v$ for at least one $v \in V_A$. By \Cref{lem:always-few-in-subtree}, the bad event happens with probability $o(1)$. For the next lemma, we condition on not having a bad event.

\begin{lemma}\label{lem:same-distribution}
    Let us condition on not having the bad event defined above. Let $F$ be the final forest found by algorithm $\mc{A}$ on a graph drawn from \yesdist{} after at most $O(n^{6/5}/\log^2 n)$ queries. Then, the probability of querying the same forest in a graph that is drawn from \nodist{} is at least almost as large, up to  $1 + o(1)$ multiplicative factor.
\end{lemma}

\begin{proof}
    First, we define $n$ hybrid distributions $\mc{D}_0, \mc{D}_1, \ldots, \mc{D}_{n-1}$ as follows. Distribution $\mc{D}_i$ can be obtained by sampling from \yesdist{} until the $i$-th level in any tree in the forest, and  then sampling from \nodist{} below the $i$-th level. Hence, $\mc{D}_0$ is exactly the same as \nodist{} and $\mc{D}_{n-1}$ is the same as \yesdist{}.
    Our goal is, starting from a forest sampled according $\mc{D}_{n-1} = \yesdist{}$, to inductively show that we can switch from $\mc{D}_i$ to $\mc{D}_{i-1}$ with only negligble total decrease in the probability.
    
    To formalize our coupling argument, recall the notion of {\em special edges} from the input distribution: In \yesdist{}, edges in $G[A]$ are special; in \nodist{}, every vertex in $A$ has one special edge to $B$ all forming a matching. We also extend the definition of bad events to \nodist{} and hybrid distribuitons to include bad events in subtrees originating from special edges.
    
    Consider the forest $F$ found by algorithm $\mc{A}$ on a graph drawn from \yesdist{}, and let $F^{\le i}$ denote the forest as well as any choice of  its special edges in all levels ${\le i}$ (arbitrarily, conditioning on not having a bad event).  We compare the probability of seeing $F^{\le i}$ when making the same queries when the oracle samples its answers from $\mc{D}_{i}$ vs $\mc{D}_{i-1}$. We can couple the sampling for the two distributions so that it is identical for everything in levels ${<i}$ (including the choice of special edges), and also for levels $\ge i$ in all the sub-trees that are not descendants of special level-$i$ edges.
    
    For each special edge $(u \rightarrow v)$ in the $i$-th level, the label of vertex $v$ is $A$ when sampling from $\mc{D}_{i}$ and  $B$ when sampling from  $\mc{D}_{i-1}$. Below this vertex (aka levels $>i$), both distributions $\mc{D}_{i}$ and $\mc{D}_{i-1}$ sample according to \nodist{}. 
    Thus, by~\cref{lem:cor-decay}, the probability of sampling the subtree below $v$ is the same regardless of the label of $v$, up to a factor of  $\left(1 \pm \frac{1}{d}\right)^{O(Q(T(v_i), \ell))}$.
    
    Let $V_{A,i}$ denote the set of vertices pointed to by $i$-th level special edges. By the argument in the previous paragraphs, we have
    \begin{align}\label{eq:special-edges}
       \Pr[F^{\le i} |  F^{\le i} \sim \mc{D}_{i-1} ] & \geq \Pr[F^{\le i} | F^{\le i} \sim \mc{D}_i] \left(1 - \frac{1}{d}\right)^{\sum_{v\in V_{A,i}} O(Q(T(v), \ell))}.
    \end{align}

    Let $\mc{G}^{\le i}$ denote the event of having no bad events corresponding to level-$(\le i)$ special edges. Summing over~\eqref{eq:special-edges} for all valid choices of special edges, we have the following:
    \begin{align*}
       \Pr[F \wedge \mc{G}^{\le i-1} |  F \sim \mc{D}_{i-1} ] & \geq \Pr[F \wedge \mc{G}^{\le i} |  F \sim \mc{D}_{i-1} ] \\
       & \geq \Pr[F \wedge \mc{G}^{\le i} | F \sim \mc{D}_i] \cdot \left(1 - \frac{1}{d}\right)^{\sum_{v\in V_{A,i}} O(Q(T(v), \ell))}.
    \end{align*}

    We can now bound  the total distribution shift across all hybrid steps:
    \begin{flalign*}
       \Pr[F |  F \sim \mc{D}_{0}] & \geq \Pr[F \wedge \mc{G}^{\le n-1} | F \sim \mc{D}_{n-1}] \cdot \prod_{i=0}^{n-1}   \left(1 - \frac{1}{d}\right)^{\sum_{v\in V_{A,i}} O(Q(T(v), \ell))}\\
       & = \Pr[F \wedge \mc{G}^{\le n-1} | F \sim \mc{D}_{n-1}] \cdot \left(1 - \frac{1}{d}\right)^{\sum_{v\in V_{A}} O(Q(T(v), \ell))} \\
       & \geq \Pr[F \wedge \mc{G}^{\le n-1} | F \sim \mc{D}_{n-1}] \cdot \left(1 - \frac{1}{d}\right)^{o(d)} \tag{By \Cref{lem:always-few-in-subtree}.} \\
       & \geq \Pr[F \wedge \mc{G}^{\le n-1} | F \sim \mc{D}_{n-1}] \cdot (1-o(1)).
    \end{flalign*}
This completes the proof of \cref{lem:same-distribution}.
\end{proof}

\begin{proof}[Proof of \Cref{thm:lb}]
    Note that the probability of having the bad event that we defined is $o(1)$. Conditioning on not having the bad event, by \Cref{lem:same-distribution}, the distribution of the outcome that the algorithm discovers is in $o(1)$ total variation distance for \yesdist{} and \nodist{}. Therefore, the algorithm is not able to distinguish between the support of two distributions with constant probability. According to our construction, the size of the maximum matching of both \yesdist{} and \nodist{} is $\Theta(n)$, and in both cases, the input graph is bipartite.
\end{proof}

\section{An (Almost) 2/3-Approximation}\label{sec:23-algo}

Throughout this section, we introduce an algorithm that achieves an almost 2/3 approximation in sublinear time in both the adjacency list and matrix models, thereby proving \cref{thm:2/3}. Specifically, we prove the following theorems, which further formalize \cref{thm:2/3} of the introduction.

\begin{restatable}{theorem}{apprxtheoremA}\label{thm:apprx23theorem}
    For any constant $\epsilon > 0$, there exists an algorithm that estimates the size of the maximum matching in $\widetilde{O}_\epsilon(n^{2-\epsilon^3})$ time up to a multiplicative-additive factor of $(\frac{2}{3} - \epsilon, o(n))$ with high probability in both adjacency matrix and adjacency list model. 
\end{restatable}

\begin{restatable}{theorem}{apprxtheoremAlist}\label{thm:apprx23theoremlist}
    For any constant $\epsilon > 0$, there exists an algorithm that estimates the size of the maximum matching in $\widetilde{O}_\epsilon(n^{2-\epsilon^3})$ time up to a multiplicative factor of $(\frac{2}{3} - \epsilon)$ with high probability in the adjacency list model. 
\end{restatable}

One key component of our algorithm is the {\em edge-degree constrained subgraph} (EDCS) of Bernstein and Stein \cite{bernsteinstein2015}. An EDCS is a maximum matching sparsifier that has been used extensively in the literature on maximum matching in different settings such as dynamic and streaming. It is known that an EDCS contains an almost 2/3-approximation of the maximum matching of the original graph. Formally, we can define EDCS as follows:

\begin{definition}[EDCS] \label{def:edcs}
For any $\lambda < 1$ and $\beta \geq 2$, subgraph $H$ is a $(\beta, (1-\lambda)\beta)$-EDCS of $G$ if
\begin{itemize}
    \item (Property P1:) for all edges $(u, v) \in H$, $\deg_H(u) + \deg_H(v) \leq \beta$, and
    \item (Property P2:) for all edges $(u, v) \in G \setminus H$, $\deg_H(u) + \deg_H(v) \geq (1-\lambda)\beta$.
\end{itemize}
\end{definition}

We use a more relaxed version of EDCS in our algorithm that first appeared in the literature on random-order streaming matching \cite{bernsteinstreaming2020}. Let $H = (V, E_H)$ be a subgraph of $G$ such that it satisfies the first property of the EDCS (P1) in \Cref{def:edcs} for some $\beta = \widetilde{O}(1)$, and $E_U$ be the set of all edges of $G$ that violate the second property of EDCS (P2). It is possible to show that the union of $E_H$ and $E_U$ contains an almost 2/3-approximation of the maximum matching of the original graph.

\begin{table}[]
\begin{center}
\caption{Subgraphs considered by algorithms in Sections~\ref{sec:23-algo} and~\ref{sec:beating}.}
\begin{tabular}{|l|l|l|}
\hline
Variable & Definition  &  Intuition \\
\hline\hline
$H$ & See \Cref{alg:matching23} & Relaxed EDCS \\
\hline
$G^{Unsampled}$ & Edges not sampled by \Cref{alg:matching23} & $\mu(G^{Unsampled}) \geq (1-2\epsilon)\mu(G)$ \\
\hline
$U$ & See \Cref{alg:matching23} & \begin{tabular}{@{}c@{}} Unsampled low $H$-degree edges \\(violating P1)\end{tabular} \\ 
\hline
$G'$ & $G' = (V, E_H \cup E_U)$ & \begin{tabular}{@{}c@{}}We can estimate $\mu(G')$, and also \\
                                                            $\mu(G')$ is a  $2/3$-approx of $\mu(G)$  \end{tabular}\\
                                                            \hline\hline
                                                            
$G''$   & $G'' =  G'[V_{low}, V_{mid}]$ & \begin{tabular}{@{}c@{}} Subgraph that contains 2/3-approx\\matching and is far from being maximal \end{tabular}\\
\hline
$G[A]$ & $A = V_{mid} \setminus V(M)$ & Edges in $G[A]$ can augment $\mu(G'')$\\
\hline
$G[V_{mid} \setminus V(M^i_{AB})]$ & $M^i_{AB}$ defined in \Cref{alg:matchingbetterthan23} & $G[V_{mid} \setminus V(M^i_{AB})]$ is non-trivially sparse\\
\hline
$G_{M_{AB}}$ & \begin{tabular}{@{}c@{}} $G_{M_{AB}} = (V_{mid}, \bigcup_i^k M_{AB}^i)$\\ defined in \Cref{alg:matchingbetterthan23} \end{tabular} & \begin{tabular}{@{}c@{}}Helps to augment $G''$ or remove \\vertices of $V_{mid} \setminus A$ in case 3 \\
\end{tabular}\\
\hline
\end{tabular}
\end{center}

\end{table}

\begin{definition}[Bounded edge-degree]
    A graph $H$ has a bounded edge-degree $\beta$ if for each edge $(u, v) \in H$, we have $\deg_H(u) + \deg_H(v) \leq \beta$.
\end{definition}

\begin{proposition}[Lemma 3.1 of \cite{bernsteinstreaming2020}, (Relaxed EDCS)]\label{pro:large-matching-union-edcs}
    Let $\epsilon \in [0, 1/2)$ and $\lambda, \beta$ be parameters such that $\lambda \leq \frac{\epsilon}{128}$, $\beta \geq \frac{16\log(1/\lambda)}{\lambda^2}$. Let $H = (V, E_H)$ be a subgraph of $G$ with bounded edge-degree $\beta$. Furthermore, let $\widetilde{H} = (V, E_U)$ be a subgraph that contains all edges $(u,v)$ in $G$ such that $\deg_H(u) + \deg_H(v) < (1-\lambda)\beta$. Then we have $\mu(H \cup \widetilde{H}) \geq (3/2-\epsilon) \cdot\mu(G)$.
\end{proposition}

\paragraph{An Informal Overview of the Algorithm:} In the next few paragraphs, we give an informal overview so the readers know what to expect. \Cref{alg:matching23} consists of two main parts. Let $H$ initially be an empty subgraph. First, in $n\beta^2 + 1$ rounds ($\beta$ is the parameter of EDCS), we sample $n^{1-\epsilon^3}$ pairs of vertices in the graph (in the adjacency list model, we sample $\delta$ edges). After sampling in each round, we update the bounded edge-degree subgraph $H$. For each edge $(u,v)$, if the $\deg_H(u) + \deg_H(v) < (1-\lambda) \beta$, we add it to $H$. This change can affect the bounded edge-degree property of $H$. In order to avoid this issue, we iterate over incident edges of $(u,v)$ to remove edges with degree more than $\beta$. Note that at most one neighbor of each $u$ and $v$ will be deleted after this iteration. Since $\beta = \widetilde{O}(1)$, this part can be done in $\widetilde{O}(n^{2-\epsilon^3})$ time. Furthermore, since we have chosen enough random edges, the number of {\em unsampled} edges in the graph that violate Property (P2) of the EDCS is relatively small. Note that when we remove edges to fix Property (P1), it is possible that an edge that was previously sampled but rejected due to the high $H$-degree, now has a lower $H$-degree and thus violates Property (P2). However, in the analysis, it will be sufficient to consider only unsampled edges that violate Property P2. More specifically, we show that the total number of unsampled violating edges is $\widetilde{O}(n^{1+\epsilon^3})$.

Note that the union of approximate EDCS $H$ and unsampled violating edges preserves an almost 2/3-approximate maximum matching of the original graph since most of the edges of the graph are unsampled and the unsampled subgraph itself, contains a large matching. Hence, we can apply \Cref{pro:large-matching-union-edcs} to the union of $H$ and all unsampled violating edges. Also, graph $G' = (V, E_H, E_U)$ has a low average degree. Hence, it remains to estimate a $(1-\epsilon)$-approximate matching of graph $G'$ that has a low average degree. For this part, we use the $\widetilde{O}_{\epsilon}(\Delta ^ {1/\epsilon^2})$ time {\em local computation algorithms} (LCA) of \citet*{levironitt} which itself builds on the sublinear time algorithm of \citet*{YoshidaYISTOC09}. In local computation algorithms, the goal is to compute a queried part of the output in sublinear time. One challenge here is that we do not have access to the adjacency list of graph $G'$ to use the algorithm of \cite{levironitt} as a black box. We slightly modify this algorithm to run in $\widetilde{O}_\epsilon(n\Delta^{1/\epsilon^2})$ with access to the adjacency matrix or the adjacency list of the original graph $G$. Furthermore, since this algorithm works with the maximum degree, we eliminate some high-degree vertices by losing an additive error. In what follows, we formalize the intuition given in previous paragraphs.

\begin{proposition}[\cite{levironitt}]\label{pre:lca-alg}
    There exists a randomized $(1-\epsilon)$-approximation local computation algorithm for maximum matching with running time $\widetilde{O}_\epsilon(\Delta^{1/\epsilon^2})$  using access to adjacency list.
\end{proposition}

\begin{lemma}\label{lem:lca-reduction}
    Let $G'$ be a subgraph of graph $G$. Also, let $\mathcal{A}$ be a local computation algorithm for maximum matching in $G'$ with running time $O(T)$ using access to the adjacency list of $G'$. There exists an algorithm with exactly the same approximation ratio that runs in $O(nT)$ time with access to the adjacency matrix or the adjacency list of $G$.
\end{lemma}
\begin{proof}
Since algorithm $\mathcal{A}$ runs in $O(T)$ time, it visits at most $O(T)$ vertices in the graph. For each of these vertices, we can simply query all their edges in $G'$ using $O(n)$ time having access to the adjacency matrix or the adjacency list of $G$. Therefore, for each vertex that $\mathcal{A}$ visits, we can construct the adjacency list of the vertex by spending $O(n)$ time.
\end{proof}

\begin{corollary}\label{cor: fast-low-deg-alg}
    Let $G'$ be a subgraph with maximum degree $\Delta$ of graph $G$. There exists a local computation algorithm for that for a given vertex $v$, determines if it is in the output of a $(1-\epsilon)$-approximate maximum matching of $G'$ with a running time $\widetilde{O}_\epsilon(n\Delta^{1/\epsilon^2})$ with access to the adjacency matrix or the adjacency list of $G$.
\end{corollary}
\begin{proof}
    The proof follows by combining \Cref{pre:lca-alg} and \Cref{lem:lca-reduction}.
\end{proof}

\begin{algorithm}[H]
\caption{A 2/3-Approximate Matching Algorithm in Adjacency List and Matrix}
\label{alg:matching23}

\textbf{Parameter:} $\epsilon$.

$\delta \leftarrow n^{1-\epsilon^3}$,  $H \leftarrow \emptyset$, $\lambda \leftarrow \frac{\epsilon}{128}$, $\beta \leftarrow \frac{16\log(1/\lambda)}{\lambda^2}$, $T \leftarrow n\beta^2 + 1$, $r \leftarrow 36 \log^3 n$.

\For{$i$ in $1 \ldots T$}{
    Sample $\delta$ pairs of vertices and query if there is an edge between them (in the adjacency list model, sample $\delta$ edges). Let $E_S$ be the set of sampled edges.

    Run the
    \Cref{alg:bernstein2020streaming} on $E_S$ for one epoch to update $H$.

    \lIf{subgraph $H$ did not change in this iteration}{\textbf{break}}
    
}

Let $E_U = \{(u, v) \lvert (u,v) \in E, \deg_H(u) + \deg_H(v) < (1-\lambda)\beta \}$.

Sample $r$ vertices $v_1, \ldots v_r$ from $V$ with replacement.

Let $X_i$ be the indicator if vertex $v_i$ is in the solution of maximum matching of graph $G' = (V, E_H \cup E_U)$ up to a multiplicative-additive factor of $(1-\epsilon)$ using the algorithm of \Cref{lem:1-eps-approx-subroutine-new}.

Let $X \leftarrow \sum_{i=1}^{r} X_i$ and $\widetilde{\mu} \leftarrow \frac{nX}{2r} - \frac{n}{2 \log n}$.

\Return $\widetilde{\mu}$.

\end{algorithm}

\begin{algorithm}[H]
\caption{Algorithm of \cite{bernsteinstreaming2020} on $E_S$ for One Epoch to Update $H$.}
\label{alg:bernstein2020streaming}
   \For{$(u, v)$ in $E_S$}{
        \If{$\deg_H(u) + \deg_H(v) < (1-\lambda)\beta$}{
            $H \leftarrow H \cup {(u,v)}$

            \For{$(u, w)$ in $N_H(u)$}{
                \If{$\deg_H(u) + \deg_H(w) > \beta$}{
                    Remove edge $(u, w)$ from $H$.
                    
                    \textbf{break.}
                }
            }

            \For{$(v, w)$ in $N_H(v)$}{
                \If{$\deg_H(u) + \deg_H(w) > \beta$}{
                    Remove edge $(u, w)$ from $H$.
                    
                    \textbf{break.}
                }
            }
        }
    }

\end{algorithm}

We use the same argument as \cite{bernsteinstreaming2020} to show that the average degree of violating unsampled edges is low. The only difference here is that in each epoch, we have $\widetilde{O}(n^{1-\epsilon^3})$ sampled edges but the algorithm of \cite{bernsteinstreaming2020} use $\widetilde{O}(n)$ sampled edges which causes us to get weaker bound on the average degree. First, we rewrite a useful lemma from \cite{bernsteinstreaming2020} that also holds in our case which shows that in one of $T$ epochs of \Cref{alg:matching23} the subgraph $H$ does not change and the algorithm breaks the loop in Line 3. 

\begin{lemma}[Lemma 4.2 of \cite{bernsteinstreaming2020}]\label{lem:edcs-moves-bound}
    Let $\beta > 2$ and $H = (V_H, E_H)$ be a subgraph with no edges at the beginning. An adversary inserts and deletes edges from $H$ with the following rules:
    \begin{itemize}
        \item delete an edge $(u,v)$ if $\deg_H(u) + \deg_H(v) > \beta$,
        \item insert an edge $(u, v)$ if $\deg_H(u) + \deg_H(v) < (1-\lambda)\beta$.
    \end{itemize}
    Then after at most $n\beta^2$ insertions and deletions, no legal move remains.
\end{lemma}

Note that in the original version of this lemma in \cite{bernsteinstreaming2020}, the constraint for insertion rule is $\deg_H(u) + \deg_H(v) < \beta - 1$, but essentially the same proof carries over without any change. Therefore, the same $n\beta^2$ bound still holds for this weaker version. By \Cref{lem:edcs-moves-bound}, in at least one of $T$ rounds of \Cref{alg:matching23}, the subgraph $H$ remains unchanged and the algorithm exits the loop in Line 3. Next, we prove that after the algorithm exits in Line 3 because of no change in $H$, the number of remaining unsampled edges that violate the second property of EDCS (P2) is small. (The proof of this lemma is similar to Lemma 4.1 of \cite{bernsteinstreaming2020}.)

Next, after the algorithm exits the loop in line 3 because of no change in $H$, the number of remaining unsampled edges that vaiolate the second property of EDCS (P2) is small. Assume that the algorithm sample $p$ fraction of edges of the original graph at random and constructs the bounded edge-degree subgraph $H$. Behnezhad and Khanna \cite{soheilsanjeev} (see Claim 4.16 of the paper) proved that the number of unsampled violating edges is at most $O(\mu(G)\cdot\beta^2 \cdot \log n/p)$. In our case, $p = O(\beta^2/n^{\epsilon^3})$ by the choice of $\delta$ and $T$. We restate the lemma with this specific $p$.

\begin{lemma}\label{lem:sparsification}
Let $\widetilde{H} = (V, E_U)$ be the subgraph consisting of violating unsampled edges in Line 7 of \Cref{alg:matching23}. Then we have $|E_U| = O(\mu(G)\cdot n^{\epsilon^3} \cdot \log n)$ with probability at least $1 - 1/n^5$.
\end{lemma}

\begin{claim}\label{clm:constructing-edcs}
    Let $E_S$ be the set of sampled edges in Line 4 of \Cref{alg:matching23}. Then running the \Cref{alg:bernstein2020streaming} for one epoch to update $H$ can be done in $O_\epsilon(n^{1-\epsilon^3})$ time. Moreover, the whole process of constructing $H$ takes $O_\epsilon(n^{2-\epsilon^3})$ time in both adjacency lists and matrix models.
\end{claim}

\begin{proof}
Since \Cref{alg:matching23} samples at most $\delta = n^{1-\epsilon^3}$ pairs of vertices, we have $|E_S| = O(n^{1-\epsilon^3})$. \Cref{alg:bernstein2020streaming} iterates over the edges one by one in each epoch and adds edge $(u,v)$ to $H$ if $\deg_H(u)+\deg_H(v) < (1-\lambda)\beta$. Adding this edge can cause incident edges of $u$ and $v$ to have a degree higher than $\beta$ in $H$. In order to remove those edges, we iterate over edges of $u$ and $v$ in $H$ and remove the first edge with a degree higher than $\beta$. This can be done in $O(\beta)$ since each vertex has at most $\beta$ neighbors. Thus, the running time of each epoch is $O(n^{1-\epsilon^3}\beta) = O_\epsilon(n^{1-\epsilon^3})$. Since we have $O(n\beta^2)$ epochs, the total running time is $O(n^{2-\epsilon^3}\beta^3) = O_\epsilon(n^{2-\epsilon^3})$ by our choice of $\beta$.
\end{proof}

\begin{lemma}\label{lem:1-eps-approx-subroutine-new}
Let $G' = (V, E_H \cup E_U)$ be as defined in \Cref{alg:matching23}. Then there exists a $(1-\epsilon)$-approximation LCA algorithm for maximum matching of $G'$ with $\widetilde{O}_\epsilon(n^{2-\epsilon^3})$ preprocessing time and $\widetilde{O}_\epsilon(n^{1+\epsilon})$ additional time per query.
\end{lemma}
\begin{proof}
Combining \Cref{lem:sparsification} and our choice of $\beta$, the average degree of graph $G'$ is $O((\mu(G) \cdot n^{\epsilon^3} \cdot \log n )/ n)$. Note that the algorithm of \cite{levironitt} works with maximum degree. In order to use this algorithm as a subroutine in Line 9 of our algorithm, we ignore vertices of high degree and find a $(1-\epsilon)$-approximate maximum matching in the remaining graph by losing an additive error which depends on $\mu(G)$. However, we do not have access to the degree of vertices in $G'$ since we do not have access to the adjacency list of $E_U$.

For each vertex $v$, we sample $k = 100n^{1-\epsilon^3}\log^2 n$ vertices to estimate the degree of $v$. Let $X$ be the number of neighbors that are in $G'$ and $\widetilde{\deg}_{G'}(v) = nX/k$ be our estimate for the degree of $v$ in $G'$. Using Chernoff bound, we get
\begin{align*}
    \Pr\left[|\widetilde{\deg}_{G'}(v) - \deg_{G'}(v) | > n^{\epsilon^3}\right] \leq 2\exp\left(-\frac{10000\log^4 n}{1500\log^3 n}\right) < \frac{2}{n^6}.
\end{align*}
A union bound over all $n$ vertices yields that with a probability of $1-1/n^4$, we have an additive error of at most $n^{\epsilon^3}$ for the degree of all vertices. 

Now we ignore the vertices with estimated degrees larger than $n^{\epsilon^3} \log^2 n$. Combining the additive error bound for degree estimation and the fact that the average degree of the graph is at most $O((\mu(G) \cdot n^{\epsilon^3} \cdot \log n )/ n)$, the total number of ignored vertices is at most $o(\mu(G))$.

For the remaining vertices, the maximum degree is at most $\widetilde{O}(n^{\epsilon^3})$. Moreover, for an edge in $G$ we can check if the degree of the edge is smaller than $(1-\lambda)\beta$ or not (the edge belongs to $G'$ or not). Thus, we can run the algorithm of \Cref{cor: fast-low-deg-alg} in $\widetilde{O}_\epsilon(n^{1 + \epsilon})$ since the maximum degree is $\widetilde{O}(n^{\epsilon^3})$ and either we have access to adjacency matrix or adjacency list of $G$. Also, note that since the additive error is $o(\mu(G))$, the approximation ratio cannot be worse than $(1-\epsilon)\mu(G') - o(\mu(G)) \geq (1-2\epsilon)\mu(G')$. Proof of claim follows from using $\epsilon/2$ as the parameter.
\end{proof}

\begin{claim}\label{clm: large-matching-union}
    $ (\frac{2}{3} - \epsilon) \cdot \mu(G) \leq \mu(G') \leq \mu(G)$.
\end{claim}
\begin{proof}
    The fact that $G'$ is a subgraph of $G$ implies the second inequality. Let $G^{Unsampled}$ be the subgraph consisting of unsampled edges. By Lemma 2.2 of \cite{bernsteinstreaming2020}, we have $\mu(G^{Unsampled}) \geq (1-2\epsilon)\mu(G)$ with high probability. To see this, we sample $\Tilde{O}(m/n^{\epsilon^3})$ fraction of edges, however, in \cite{bernsteinstreaming2020}, they sample $\epsilon m$ edges. Hence, the graph $G^{Unsampled}$ has more unsampled edges and therefore a larger matching compared to unsampled edges of \cite{bernsteinstreaming2020}. Moreover, by \Cref{pro:large-matching-union-edcs}, we have $\mu(G') \geq (\frac{2}{3}-\epsilon) \mu(G^{Unsampled})$. Therefore, we obtain
    \begin{align*}
        \mu(G') \geq (\frac{2}{3}-\epsilon) \mu(G^{Unsampled}) \geq (\frac{2}{3}-\epsilon) (1-2\epsilon)\mu(G) \geq (\frac{2}{3} - 3\epsilon) \mu(G).
    \end{align*} 
    Proof of claim follows from using $\epsilon/3$ as the parameter.
\end{proof}

\begin{lemma}\label{lem:approx-ratio-32}
    let $\widetilde{\mu}$ be the output of \Cref{alg:matching23} on graph $G$. With high probability,
    $$
    (\frac{2}{3} - \epsilon) \mu(G) - o(n) \leq \widetilde{\mu} \leq \mu(G).
    $$
\end{lemma}
\begin{proof}
Let $\hat{M}$ be the $(1-\epsilon)$-approximate matching of \Cref{lem:1-eps-approx-subroutine-new} on graph $G'$ and $X_i$ be the indicator if vertex $i$ is matched in $\hat{M}$ or not. By \Cref{clm: large-matching-union}, we have 
\begin{align}\label{eq:bounds-on-M}
    (1-\epsilon)(\frac{2}{3} - \epsilon) \mu(G)  \leq \E|\hat{M}| \leq \mu(G).
\end{align}
Because the number of matching edges is half of the matched vertices,
\begin{align*}
    \E[X_i] = \Pr[X_i = 1] = \frac{2\E|\hat{M}|}{n}.
\end{align*}
Thus,
\begin{align}\label{eq:expected-matching}
    \E[X] = \frac{2r\E|\hat{M}|}{n}.
\end{align}
Using Chernoff bound and the fact that $X$ is the sum of $r$ independent Bernoulli random variables, we have
\begin{align*}
    \Pr[|X - \E[X]| \geq \sqrt{18\E[X]\log n}] \leq 2 \exp\left(-\frac{18\E[X]\log n}{3\E[X]}\right) = \frac{2}{n^6}.
\end{align*}
Hence, with probability $1-2/n^6$,
\begin{align*}
    \widetilde{\mu} = \frac{nX}{2r} - \frac{n}{2\log n} & \in \frac{n(\E[X] \pm \sqrt{18\E[X]\log n})}{2r} - \frac{n}{2\log n} \\
    & =  \E|\hat{M}| \pm \sqrt{\frac{9n\E|\hat{M}|\log n}{r}}- \frac{n}{2\log n} & (\text{By } \Cref{eq:expected-matching})\\
    & = \E|\hat{M}| \pm \sqrt{\frac{n\E|\hat{M}|}{4\log^2 n}}- \frac{n}{2\log n} & (\text{Since } r = 36 \log^3 n)\\ 
    & = \E|\hat{M}| \pm \frac{n}{2\log n}- \frac{n}{2\log n} & (\text{Since } \E|\hat{M}| \leq n).
\end{align*}
Plugging \Cref{eq:bounds-on-M} in the above range implies
\begin{align*}
    (1-\epsilon)(\frac{2}{3} - \epsilon) \mu(G) - \frac{n}{\log n} \leq \widetilde{\mu} \leq \mu(G).
\end{align*}
Since $\frac{n}{\log n} = o(n)$ and $(1-\epsilon)(\frac{2}{3} - \epsilon) > (\frac{2}{3} - 3\epsilon)$, we get
\begin{align*}
    (\frac{2}{3} - 3\epsilon) \mu(G) - o(n) \leq \widetilde{\mu} \leq \mu(G).
\end{align*}
Proof of lemma follows from using $\epsilon/3$ as the parameter.
\end{proof}

Now we are ready to complete the analysis of the 2/3-approximate maximum matching algorithm.

\apprxtheoremA*

\begin{proof}    
    We run \Cref{alg:matching23}. By \Cref{lem:approx-ratio-32}, we obtain the claimed approximation ratio. The proof of running time follows from combining \Cref{clm:constructing-edcs}, \Cref{lem:1-eps-approx-subroutine-new}, and $r = \widetilde{O}(1)$.
\end{proof}

\paragraph{Multiplicative Approximation for Adjacency List:} In the adjacency list model, we can assume that we are given the degree of each vertex. This assumption is without a loss of generality since we can use binary search for each vertex to find its exact degree. Thus, we know the total number of edges in the graph and if this number is not larger than $n^{1.99}$, then we can use linear time $1-\epsilon$ approximation for maximum matching. Equipped with this observation, we assume that $\mu(G) \geq n^{0.99} / 2$, otherwise, the number of edges cannot be more than $2n\cdot \mu(G)$. Moreover, the size $\mu(G')$ is at least $(2/3 - \epsilon)\mu(G)$, which implies that we can assume $\mu(G') = \Omega(n^{0.99})$. With this assumption and using standard Chernoff bound, it is not hard to see that we can estimate the size of the maximum matching of $G'$ with a multiplicative factor by sampling $r = \widetilde{O}(n^{0.01})$ vertices. By the running time of \Cref{clm:constructing-edcs}, preprocessing time of \Cref{lem:1-eps-approx-subroutine-new}, and query time of \Cref{lem:1-eps-approx-subroutine-new} combined with $r$ queries, results in multiplicative $(2/3 - \epsilon)$-approximation algorithm with $\widetilde{O}_\epsilon(n^{2-\epsilon^3})$ running time.

\apprxtheoremAlist*

\section{Beating 2/3-Approximation in Bipartite Graphs}\label{sec:beating}

In this section, we design a new algorithm that gets slightly better than $2/3$ approximation in sublinear time in both the adjacency list and matrix models for bipartite graphs. Our starting point is to use the tight case characterization of the instance that our algorithm in the previous part cannot obtain better than a $2/3$ approximation. Then we use that characterization to design a new algorithm to go beyond $2/3$. We prove the following theorem, which further formalizes the statement of \cref{thm:beating-2/3} in the introduction.

\begin{restatable}{theorem}{apprxtheoremB}\label{thm:apprxbetter23theorem}
    For an absolute constant $\alpha$, there exists an algorithm that estimates the size of the maximum matching in $\widetilde{O}(n^{2-\Omega(1)})$ time up to a multiplicative-additive factor of $(\frac{2}{3} + \alpha, o(n))$ with high probability in both adjacency list and matrix models.
\end{restatable}

\begin{restatable}{theorem}{apprxtheoremBlist}\label{thm:apprxbetter23theoremlist}
    For an absolute constant $\alpha$, there exists an algorithm that estimates the size of the maximum matching in $\widetilde{O}(n^{2-\Omega(1)})$ time up to a multiplicative factor of $(\frac{2}{3} + \alpha)$ with high probability in the adjacency list model. 
\end{restatable}

To break $2/3$-approximation, we build on a characterization of tight instances of EDCS due to a recent work of \citet*{behnezhad2022dynamic}. In our context, the characterization asserts that if the output of our \Cref{alg:matching23} is not already a $(2/3 + \alpha)$-approximate matching of $G$ for some constant $\alpha \ge  10^{-24}$ (we did not try to optimize the constant), then there exists a $2/3$-approximate matching in $G'$ (defined in \Cref{alg:matching23}) such that it is far from being maximal in the original graph $G$.

\paragraph{Notation} 
Let $G$ be the input bipartite graph, and let $G' = (V, E_H \cup E_U)$ and $H = (V, E_H)$ be as defined in \Cref{alg:matching23}. Let $V_{low}$ and $V_{mid}$ be the set of vertices $v$ such that $\deg_H(v) \in [0, 0.2\beta]$ and $\deg_H(v) \in [0.4\beta, 0.6\beta]$, respectively. Let $G'' \coloneqq G'[V_{low}, V_{mid}]$.
Also, fix $M$ to be a $(1-\epsilon)$-approximate matching of $G'' \coloneqq G'[V_{low}, V_{mid}]$.
We let $A \coloneqq V_{mid} \setminus V(M)$ and $B \coloneqq V_{mid} \cap V(M)$ for the rest of this section.

\begin{lemma}[\cite{behnezhad2022dynamic}]\label{lem:tight-instance-edcs}
    Let $\epsilon < 1/120$  and $\alpha$ be some constant. If $\mu(G') \leq (\frac{2}{3} + \alpha)\mu(G)$, then both of the following guarantees hold:
    \begin{itemize}
        \item {\bf (G1 - $G''$ contains a large matching $M$)} $\mu(G'') \geq (\frac{2}{3} - 120 \sqrt{\alpha}) \cdot\mu(G),$ 
        \item {\bf (G2 -  $M$ can be augmented in $G[A]$)} For any matching $M$ in $G''$, we have $\mu(G[V_{mid} \setminus V(M)]) \geq (\frac{1}{3} - 800 \alpha)\cdot \mu(G).$
    \end{itemize}
\end{lemma}

\paragraph{An Informal Overview of the Algorithm:} Let $M$ be a $(1-\epsilon)$-approximate matching of $G'' \coloneqq G'[V_{low}, V_{mid}]$. According to this characterization, if \Cref{alg:matching23} does not obtain better than a $(2/3 + \alpha)$ approximation, then matching $M$ preserves the size of matching of $G'$ and also, it can be augmented using a matching with edges of $G[V_{mid} \setminus V(M)]$.  By \Cref{lem:tight-instance-edcs}, if the matching that we estimate in \Cref{alg:matching23} does not obtain better than $2/3$ approximation, then there exists a large matching in $G''$ that can be augmented using a matching between vertices of $A$.

Therefore, if we want to obtain a better than $2/3$ approximation, all we need is to estimate the size of a constant fraction of a matching in $G[A]$. The first challenge is that we do not know which vertices of $V_{mid}$ are in $A$ and which are in $B$. With a similar proof as proof of \Cref{lem:1-eps-approx-subroutine-new}, we can show that it is possible to query if a vertex is matched in $M$ in $\widetilde{O}_\epsilon(n^{1+\epsilon})$ time, since $G''$ is a subgraph of $G'$ and its average degree is smaller than $G'$. Moreover, for an edge in $G$, we can check if the $H$-degree of the edge is smaller than $(1-\lambda)\beta$ (the edge belongs to $G'$ or not), and one endpoint is in $V_{mid}$ and the other one in $V_{low}$ (the edge belongs to $G''$ or not) in $O(1)$. Thus, we can run the algorithm of \Cref{cor: fast-low-deg-alg}. We restate the lemma for subgraph $G''$.

\begin{lemma}\label{lem:mid-low-lca}
    There exists a $(1-\epsilon)$-approximation LCA algorithm for maximum matching of $G'' \coloneqq G'[V_{low}, V_{mid}]$ with $\widetilde{O}_\epsilon(n^{2-\epsilon^3})$ preprocessing time and $\widetilde{O}_\epsilon(n^{1+\epsilon})$ additional time per query.
\end{lemma}

Since $G'[V_{low}, V_{mid}]$ is a subgraph of $G'$, by \Cref{lem:1-eps-approx-subroutine-new}, we can query if a vertex is matched in $M$ in $\widetilde{O}_\epsilon(n^{1+\epsilon})$. 
So the time to classify if a vertex is in $A$ or $B$ is $\widetilde{O}_\epsilon(n^{1+\epsilon})$ since we can query if it is matched in matching $M$ or not. 

Unfortunately, the subgraph $G[A\cup B]$ can be a dense graph (possibly with an average degree of $\Theta(n)$), so it is not possible to run an adaptive sublinear algorithm that adaptively queries if a vertex is in $A$ or not while looking for a matching. So we first try to sparsify the graph by sampling edges and constructing several maximal matchings. More specifically, we query $\delta k$ pairs of vertices (for $\delta = n^{1+\gamma}$ and $k = n^{\gamma/2}$) and partition the queries to $k$ equal buckets. Then we build a greedy maximal matching $M_{AB}^i$ among those sampled edges of bucket $i$. By this construction and the sparsification property of greedy maximal matching, we are able to show that the maximum degree of subgraph $G[V_{mid} \setminus V(M^i_{AB})]$ is $O(n^{1-\gamma})$ with high probability for each $i$.

We split the rest of the analysis into three possible cases.

    \paragraph{(Case 1) A constant fraction of vertices of $A$ are matched in at least one of $M_{AB}^i$:} in this case we immediately beat $2/3$ since we found a constant fraction of a maximum matching edges of $G[A]$.

    \paragraph{(Case 2) A constant fraction of maximum matching of $G[A]$ has both endpoints unmatched in at least one of $M_{AB}^i$:} Let $i^*$ be an index of matching $M_{AB}^{i^*}$ that a constant fraction of maximum matching of $G[A]$ has both endpoints unmatched. In this case, we prove that it is possible to estimate a constant approximation of the size of the matching between unmatched vertices of $A$ in sublinear time. 
    
    Let $G_{AB}^{i^*}$ be the same as subgraph $G[V_{mid} \setminus V(M_{AB}^{i^*})]$, except that we connect each vertex in $B \setminus V(M_{AB}^{i^*})$ to $\kappa = \widetilde{O}(n^{1-\gamma})$ dummy singleton vertices. The reason to connect singleton vertices to $B$ is that if we run a random greedy maximal matching on $G_{AB}^{i^*}$, most of the vertices in $B$ will match to singleton vertices with some additive error. Thus, if there exists a large matching between unmatched $A$ vertices, random greedy maximal matching will find at least half its edges. We sample $\widetilde{O}(1)$ vertices and test if they are matched in the random greedy maximal matching of $G_{AB}^{i^*}$. To do this, we use the algorithm of \cite{behnezhad2021} which has a running time of $\widetilde{O}(n^{1-\gamma})$ for a random vertex. However, for each vertex that this algorithm recursively makes a query, we have to classify if it is in $A$ or $B$ since we do not explicitly construct $G_{AB}^{i^*}$. As we discussed at the beginning of this overview, this task can be done in $\widetilde{O}_\epsilon(n^{1+\epsilon})$ time. For each maximal matching, the running time for this case is $\widetilde{O}_\epsilon(n^{2+\epsilon-\gamma})$. Therefore, the total running time is $\widetilde{O}_\epsilon(n^{2+\epsilon-\gamma/2})$ for all $n^{\gamma/2}$ maximal matchings which is sublinear if we choose $\gamma$ sufficiently larger than $\epsilon$.

Note that if we are not in one of the above cases, almost all edges of $n^{\gamma/2}$ maximal matchings have at most one endpoint in $A$ (as a result of not being in Case 1), and almost all edges of the maximum matching of $G[A]$ have at least one endpoint matched by all $n^{\gamma/2}$ maximal matchings (as a result of not being in Case 2). Thus, if we consider a vertex in $A$, almost all of its incident maximal matching edges go to set $B$. For this case, we now describe how it is possible to estimate the matching of $G[A]$.

    \paragraph{(Case 3) Almost every edge of maximum matching of $G[A]$ has at least one of its endpoints matched by almost all $M_{AB}^i$:} 
    In this case we will show that we can efficiently remove many $B$-vertices, and then repeat the algorithm and and analysis of Cases 1 and 2 (using fresh samples of $M^i_{AB}$. Because we can't keep removing $B$-vertices indefinitely, eventually we have to end up at either Case 1 or 2.
    
    Our goal is henceforth to efficiently identify many $B$-vertices. Suppose that we sample a vertex and check if it is in $A$ or $B$. If the sampled vertex is in $A$ and is one of the endpoints of maximum matching of $A$, most of its incident edges in $M_{AB}^1, M_{AB}^2, \ldots, M_{AB}^k$ are connected to vertices of $B$. Therefore, we are able to find $k = \Theta(n^{\gamma/2})$ vertices of the set $B$ by spending $\widetilde{O}_\epsilon(n^{1+\epsilon})$ to find a vertex in $A$. Hence, if we sample $\widetilde{\Theta}(n^{1-\gamma/2})$ such vertices, we are expecting to see a constant fraction of vertices of $B$; we can remove these vertices and run the same algorithm on the remaining graph again. One challenge that arises here is that some of the neighbors of the sampled vertex might be in $A$. To resolve this issue, we choose a threshold $\eta$ and only remove vertices that have at least $\eta$ maximal matching edges that are connected to sampled $A$ vertices. 
    This will help us to avoid removing many $A$ vertices in each round. Since initially the size of $B$ and $A$ is roughly the same (up to a constant factor), after a few rounds, we can estimate the size of the maximum matching of $A$. The total running time of this case is also $\widetilde{O}_\epsilon(n^{2+\epsilon-\gamma/2})$ which is the same as the previous case.

In what follows, we present the formal algorithm (\Cref{alg:matchingbetterthan23}) and formalize the technical overview given in previous
paragraphs.

\begin{algorithm}[]
\caption{Better Than 2/3-Approximate Matching Algorithm in Adjacency List and Matrix}
\label{alg:matchingbetterthan23}
\textbf{Parameter:} $\epsilon, \gamma$.

$r_1 \leftarrow 72 \log^3 n, r_2 \leftarrow 288 \log^3 n, r_3 \leftarrow 10n^{1-\gamma/2}\log n$.

$T \leftarrow 200, \alpha \leftarrow 10^{-10}, \kappa \leftarrow 48n^{1-\gamma}\log^2 n, \delta \leftarrow n^{1+\gamma}, k \leftarrow n^{\gamma/2}$, $c \leftarrow 10^{16}$.

Let $\widetilde{\mu}$ be the output of \Cref{alg:matching23} on $G$ with parameter $\epsilon$. Also, let $\beta$, $H$, $E_U$, and $G' = (V, E_H \cup E_U)$ be as defined in \Cref{alg:matching23}.

$\eta \leftarrow \frac{n\log n}{100\widetilde{\mu}_1}$.

Let $V_{low} = \{v | \deg_H(v) \in [0, 0.2\beta] \}$, $V_{mid} = \{v | \deg_H(v) \in [0.4\beta, 0.6\beta] \}$, and $G'' \coloneqq G'[V_{low}, V_{mid}]$.

Let $\widetilde{\mu}_1$ be an estimate with an additive error of $o(n)$ for the size of $(1-\epsilon)$-approximate matching of $G''$ by sampling $\widetilde{O}(1)$ vertices and running the algorithm of \Cref{lem:mid-low-lca}.

\For{$j$ in $1 \ldots T$}{
    Sample $k\delta$ pairs of vertices of $V_{mid}$ and partition them into $k$ buckets of equal size. Let $M_{AB}^i$ be the greedy maximal matching on the sampled edges of bucket $i$.

    \For{$i$ in $1 \ldots k$}{

        \textbf{// Beginning of case 1}
        
        Sample $r_1$ random edges $e_1, \ldots, e_{r_1}$ of $M_{AB}^i$.

        Let $X_{\ell}$ be the indicator if neither endpoint of $e_{\ell}$ is matched in the $(1-\epsilon)$-approximate matching of $G''$ by running algorithm of \Cref{lem:mid-low-lca}. 

        Let $X \leftarrow \sum^{r_1}_{\ell=1} X_{\ell}$ and $\widetilde{\mu}_2 = \frac{nX}{r_1} - \frac{n}{2\log n}$.

        \lIf{$\widetilde{\mu}_2 \geq c\alpha \widetilde{\mu}_1$}{
            \Return $\widetilde{\mu}_1 + \widetilde{\mu}_2$
        }

        \textbf{// End of case 1}

        \textbf{// Beginning of case 2}
        
        Let $G_{AB}^{i}$ be the same as subgraph $G[V_{mid} \setminus V(M_{AB}^{i})]$, except that we connect each vertex in $B \setminus V(M_{AB}^{i})$ to $\kappa$ singleton vertices.

        Sample $r_2$ random vertices $v_1, \ldots, v_{r_2}$ of vertex set $A \setminus V(M_{AB}^{i})$.

        Let $Y_{\ell}$ be the indicator if $v_{\ell}$ is matched in the random greedy maximal matching of $G_{AB}^i$ (algorithm of \cite{behnezhad2021}).

        Let $Y \leftarrow \sum^{r_2}_{\ell=1} Y_{\ell}$ and $\widetilde{\mu}_3 = \frac{nY}{2r_2} - \frac{n}{2\log n}$.

        \lIf{$\widetilde{\mu}_3 \geq c\alpha \widetilde{\mu}_1$}{
            \Return $\widetilde{\mu}_1 + \widetilde{\mu}_3$
        }

        \textbf{// End of case 2}
    }

    \textbf{// Beginning of case 3}
        
    Sample $r_3$ vertices $u_1, u_2, \ldots u_{r_3}$ of vertex set $A$.

    Define $G_{M_{AB}} = (V_{mid}, \bigcup_i^k M_{AB}^i)$.

    Let $C \subseteq V_{mid}$ be the set of vertices that has at least $\eta$ neighbors in $G_{M_{AB}}$ among $r_3$ sampled vertices of $A$.
    \algorithmiccomment{$C$ vertices likely in $B$}

    $V_{mid} \leftarrow V_{mid} \setminus C$. \algorithmiccomment{Remove those vertices} \label{Line:remove}

    \textbf{// End of case 3}
}

\Return $\widetilde{\mu}$

\end{algorithm}

\subsection{Analysis of Running Time}

We analyze each part of \Cref{alg:matchingbetterthan23} separately in this section and at the end, we put everything together. Before starting the analysis of the running time of \Cref{alg:matchingbetterthan23}, we first restate the following result on the time complexity of estimating the size of the maximal matching by Behnezhad \cite{behnezhad2021}.

\begin{lemma}[Lemma 4.1 of \cite{behnezhad2021}]\label{lem:RGMM-query-complexity}
    There exists an algorithm that draws a random permutation over the edges of the graph and for an arbitrary vertex $v$ in graph $G$, it determines if $v$ is matched in the random greedy maximal matching of $G$ in $\widetilde{O}(\bar{d})$ time with high probability, using adjacency list of $G$, where $\bar{d}$ is the average degree of $G$.
\end{lemma}

\begin{claim}\label{clm:fast-maximal-matching}
    Let $M_{AB}^i$ be as defined in \Cref{alg:matchingbetterthan23}. The running time for constructing $M_{AB}^i$ for all $1 \leq i\leq k$, takes $O(n^{1+3\gamma/2})$ time for all $T$ iterations of the algorithm.
\end{claim}
\begin{proof}
    Fix an iteration in the algorithm. Since we sample $\delta k = n^{1+3\gamma/2}$ pairs of vertices and we have $k$ buckets of equal size, each bucket can have at most $\delta$ edges. For each of the buckets, in order to construct the maximal matching we need to iterate over the edges of the bucket one by one which takes $O(\delta k)$ time for all buckets. The proof follows from $T = O(1)$.
\end{proof}

\begin{claim}\label{clm:preprocessing-time}
    The total preprocessing time before all iterations of \Cref{alg:matchingbetterthan23} is $\widetilde{O}_\epsilon(n^{2-\epsilon^3})$.
\end{claim}
\begin{proof}
    \Cref{alg:matchingbetterthan23} uses subroutine of \Cref{lem:mid-low-lca} to query if a vertex is matched in $(1-\epsilon)$-approximate matching of $G''$ during its execution. Thus, by \Cref{lem:mid-low-lca}, the total preprocessing time needed for the algorithm is $\widetilde{O}_\epsilon(n^{2-\epsilon^3})$.
\end{proof}

\begin{claim}\label{clm:maximal-matching-aa-edges}
    Every time that \Cref{alg:matchingbetterthan23} calculates $\widetilde{\mu}_2$, it takes $\widetilde{O}_\epsilon(n^{1+\epsilon})$ time. Furthermore, the total running time for computing $\widetilde{\mu}_2$ in the whole process of algorithm is $\widetilde{O}_\epsilon(n^{1+\epsilon + \gamma/2})$ time.
\end{claim}

\begin{proof}
    By \Cref{lem:mid-low-lca}, each query to see if a vertex is matched in $(1-\epsilon)$-approximate matching of $G''$, takes $\widetilde{O}_\epsilon(n^{1+\epsilon})$ time. Since we sample $\widetilde{O}(1)$ vertices, we get the claimed time complexity.
\end{proof}

\begin{lemma}\label{lem: sparsification-property}
    Subgraph $G[V_{mid} \setminus V(M^i_{AB})]$ has a maximum degree of $6n^{1-\gamma}\log n$ with high probability.
\end{lemma}

\begin{proof}
    Consider a vertex $v$ that is not matched in $M^i_{AB}$ and has a degree larger than $6n^{1-\gamma}\log n$ in $G[V_{mid} \setminus V(M^i_{AB})]$. The only way that $v$ remains unmatched is that none of its incident edges in $G[V_{mid} \setminus V(M^i_{AB})]$ get sampled in $n^{1+\gamma}$ pairs of vertices we sampled. The probability that we sample one of these edges is at least $6n^{1-\gamma}\log n/n^2$, since there are at most $n^2$ possible pairs of vertices. Thus, the probability that none of these edges sampled in $n^{1+\gamma}$ samples when constructing $M_{AB}^i$ is at most
    \begin{align*}
        \left( 1 - \frac{6n^{1-\gamma}\log n}{n^2} \right)^{n^{1 + \gamma}} = \left( 1 - \frac{6\log n}{n^{1+\gamma}} \right)^{n^{1 + \gamma}} \leq n^{-6}.
    \end{align*}
    Taking a union bound over all vertices of the graph completes the proof.
\end{proof}

\begin{claim}\label{clm:estimating-aa-missing-matching}
    Every time that \Cref{alg:matchingbetterthan23} calculates $\widetilde{\mu}_3$, it takes $\widetilde{O}_\epsilon(n^{2+\epsilon-\gamma})$ time. Furthermore, the total running time for computing $\widetilde{\mu}_3$ in the whole process of algorithm is $\widetilde{O}_\epsilon(n^{2+\epsilon-\gamma/2})$ time.
\end{claim}

\begin{proof}
    By \Cref{lem: sparsification-property}, the maximum degree of subgraph $G_{AB}^i$ is $\widetilde{O}(n^{1-\gamma})$ with high probability. Thus, by \Cref{lem:RGMM-query-complexity}, indicating if a vertex is matched in the random greedy maximal matching of $G^i_{AB}$ takes $\widetilde{O}(n^{1-\gamma})$ time. Moreover, for each vertex in the process of running a random greedy maximal matching local query algorithm, we need to distinguish if it is a vertex in $A$ or $B$ since we build subgraph $G^i_{AB}$ on the fly (we cannot afford to build the whole graph at the beginning which takes $\Theta(n^2)$ time). By \Cref{lem:mid-low-lca}, each of these queries takes $\widetilde{O}_\epsilon(n^{1+\epsilon})$ time. Therefore, it takes $\widetilde{O}_\epsilon(n^{2+\epsilon-\gamma})$ time to calculate $\widetilde{\mu}_3$ each time. Combining with the fact that $r_2 = \widetilde{O}(1)$, we have $O(n^{\gamma/2})$ maximal matchings, the size of set $A$ is a constant fraction of $V_{mid}$, and $T = O(1)$ implies the claimed running time.
\end{proof}

\begin{claim}\label{clm:updating-vmid}
    Updating vertex set $V_{mid}$ at the end of case 3 in \Cref{alg:matchingbetterthan23} takes $\widetilde{O}_\epsilon(n^{2+\epsilon-\gamma})$ time. Furthermore, the total running time for computing $V_{mid}$ in the whole process of algorithm is $\widetilde{O}_\epsilon(n^{2+\epsilon-\gamma})$ time.
\end{claim}

\begin{proof}
    Each time that we sample a vertex, we need to test if it is a vertex in $A$. This takes $\widetilde{O}_\epsilon(n^{1+\epsilon})$ by \Cref{lem:mid-low-lca}. Since we need to sample $r_3 = \widetilde{O}(n^{1-\gamma/2})$ vertices of $A$ and the size of set $A$ is a constant fraction of $V_{mid}$ by \Cref{lem:tight-instance-edcs}, we need to sample at most $\widetilde{O}(n^{1-\gamma})$ vertices of $V_{mid}$. Hence, it takes $\widetilde{O}_\epsilon(n^{2+\epsilon-\gamma/2})$ time to find vertices $u_1, u_2, \ldots, u_k \in A$. Each of the vertices has at most $k$ neighbors in all maximal matchings $M_{AB}^1, M_{AB}^2 \ldots, M_{AB}^{k}$ which implies that we need to spend $\widetilde{O}(n)$ time to find and remove vertices of set $C$ which completes the proof since $T = O(1)$.
\end{proof}

Now we are ready to complete the analysis of the running time of \Cref{alg:matchingbetterthan23}.

\begin{lemma}\label{lem:time-better-23}
    The total running time of \Cref{alg:matchingbetterthan23} is $\widetilde{O}_\epsilon(\max(n^{2+\epsilon-\gamma/2}, n^{2-\epsilon^3}, n^{1+3\gamma/2}, n^{1+\epsilon+\gamma/2}))$.
\end{lemma}

\begin{proof}
    The proof follows from combining \Cref{lem:mid-low-lca}, \Cref{clm:fast-maximal-matching}, \Cref{clm:preprocessing-time}, \Cref{clm:maximal-matching-aa-edges}, \Cref{clm:estimating-aa-missing-matching}, and \Cref{clm:updating-vmid}.
\end{proof}

\subsection{Analysis of Approximation Ratio}

In this section, we assume that $\widetilde{\mu} \leq (\frac{2}{3} + \alpha)\cdot\mu(G)$ and we prove with this assumption, the algorithm outputs a better than $2/3$ approximation in one of $T$ iterations. The proof for the other case where $\widetilde{\mu} > (\frac{2}{3} + \alpha)\cdot\mu(G)$ is trivial since the algorithm outputs an estimation which is at least $\widetilde{\mu}$.

\subsubsection{Approximation analysis of Case 1: Many $A$-vertices matched in some $M_{AB}^i$}

\begin{lemma}\label{lem:case1-approx}
    If $\widetilde{\mu}_2 \geq c\alpha\widetilde{\mu}_1$, then $\widetilde{\mu}_1 + \widetilde{\mu}_2$ is an estimate for $\mu(G)$ up to a multiplcative-additive factor of $(\frac{2}{3} + \alpha, o(n))$ with high probability.
\end{lemma}

\begin{proof}
First, we prove that $\widetilde{\mu}_1 + \widetilde{\mu}_2 \leq \mu(G)$. Let $X$ be as defined in \Cref{alg:matchingbetterthan23}. Note that $\frac{n\E[X]}{r_1} \leq \mu(G[A])$ since it counts the number of edges in $M_{AB}^i$ that have both of their endpoints in $A$. Since $X$ is the sum of $r_1$ independent Bernoulli random variables, using Chernoff bound we have
\begin{align*}
    \Pr[|X - \E[X]| \geq \sqrt{18\E[X]\log n}] \leq 2 \exp\left(-\frac{18\E[X]\log n}{3\E[X]}\right) = \frac{2}{n^6}.
\end{align*}
Thus, with a probability of at least $1-2/n^6$,
\begin{align*}
    \widetilde{\mu}_2 \leq \frac{nX}{r_1} - \frac{n}{2\log n} & \leq \frac{n\left(\E[X] + \sqrt{18\E[X]\log n}\right)}{r_1} - \frac{n}{2\log n} \\
    & \leq \mu(G[A]) + \sqrt{\frac{18n\cdot\mu(G[A])\log n}{r_1}} - \frac{n}{2\log n} & (\text{Since } \frac{n\E[X]}{r_1}\leq \mu(G[A]))\\
    & \leq \mu(G[A]) + \sqrt{\frac{n\cdot \mu(G[A])}{4\log^2 n}} - \frac{n}{2\log n} & (\text{Since } r_1 = 72 \log^3 n)\\
    & \leq \mu(G[A]) & (\text{Since } \mu(G[A]) \leq n).
\end{align*}
Moreover, we have that $\widetilde{\mu}_1 \leq \mu(G'')$. Since $G[A]$ only contains vertices that are not matched in the matching that we found in $G''$, we get $\widetilde{\mu}_1 + \widetilde{\mu}_2 \leq \mu(G)$. On the other hand, by \Cref{lem:tight-instance-edcs}, we have $\widetilde{\mu}_1 \geq (1-\epsilon)\cdot (\frac{2}{3} - 120\sqrt{\alpha})\mu(G) - o(n)$. Thus,
\begin{align*}
    \widetilde{\mu}_1 + \widetilde{\mu}_2 \geq (1+c\alpha)\widetilde{\mu}_1 & \geq (1+c\alpha)\left((1 - \epsilon)\left(\frac{2}{3}-120\sqrt{\alpha}\right)\cdot\mu(G) - o(n)\right) \\
    & \geq (1+c\alpha)\cdot \left(\frac{2}{3} - 120\sqrt{\alpha} - \frac{2}{3}\epsilon\right)\cdot \mu(G) - o(n)\\
    & \geq \left(\frac{2}{3} - 120\sqrt{\alpha} - \frac{2}{3}\epsilon +\frac{2}{3}c\alpha - 120c\alpha\sqrt{\alpha} - \frac{2}{3}c\alpha \epsilon \right)\cdot\mu(G) - o(n) \\
    & \geq \left(\frac{2}{3} + \alpha \right)\mu(G) - o(n),
\end{align*}
where the last inequality holds by our choice of $c$, $\alpha$, and if we choose $\epsilon$ sufficiently small enough, which leads to our claimed approximation ratio.
\end{proof}

\subsubsection{Approximation analysis of Case 2: Large matching in $G[A \setminus M^i_{AB}]$ (for some $M^i_{AB}$)}

\begin{lemma}\label{lem:case2-approx}
    If $\widetilde{\mu}_3 \geq c\alpha\widetilde{\mu}_1$, then $\widetilde{\mu}_1 + \widetilde{\mu}_3$ is an estimate for $\mu(G)$ up to a multiplcative-additive factor of $(\frac{2}{3} + \alpha, o(n))$ with high probability.
\end{lemma}

\begin{proof}
Note that since we connect each vertex of $B \setminus V(M_{AB}^{i})$ to $\kappa = \widetilde{O}(n^{1-\gamma})$ singleton vertices, most of the vertices of $B \setminus V(M_{AB}^{i})$ will be matched to singleton vertices when running random greedy maximal matching. Let $v \in B \setminus V(M_{AB}^{i})$. The first incident edge of $v$ in the permutation of edges is not an edge to singleton vertices with a probability of at most $(6n^{1-\gamma}\log n) / (\kappa + 6n^{1-\gamma}\log n)$ since the maximum degree of subgraph $G[V_{mid} \setminus V(M^i_{AB})]$ is at most $6n^{1-\gamma}\log n$ by \Cref{lem: sparsification-property}. We denote the vertices of $B \setminus V(M^i_{AB})$ that are not matched with singleton vertices by $R$. Hence, we have the following bounds,
\begin{align*}
    \E|R| \leq n\cdot\left(\frac{6n^{1-\gamma}\log n}{\kappa + 6n^{1-\gamma}\log n}\right) \leq n\cdot\left(\frac{6n^{1-\gamma}\log n}{\kappa}\right) = \frac{n}{8\log n},\\
    \E|R| \geq n\cdot\left(\frac{6n^{1-\gamma}\log n}{\kappa + 6n^{1-\gamma}\log n}\right) \geq n\cdot\left(\frac{6n^{1-\gamma}\log n}{2\kappa}\right) = \frac{n}{16\log n}.
\end{align*}
Therefore, using a Chernoff bound, we can show that $|R| \leq \frac{n}{4\log n}$ with a high probability.

Now, we prove that $\widetilde{\mu}_1 + \widetilde{\mu}_3 \leq \mu(G)$. Let $Y$ be as defined in \Cref{alg:matchingbetterthan23}, then we have $\frac{n\E[Y]}{2r_2} \leq \mu(G[A]) + |R|$, since it counts the number of edges in output of random greedy maximal matching with both endpoints in $A$ and there are at most $|R|$ vertices of $B$ can match to vertices of $A$. Also, each edge is counted at most twice since it has two endpoints in $A$. Thus,
\begin{align*}
    \Pr[|Y - \E[Y]| \geq \sqrt{18\E[Y]\log n}] \leq 2 \exp\left(-\frac{18\E[Y]\log n}{3\E[Y]}\right) = \frac{2}{n^6},
\end{align*}
since $Y$ is the sum of independent Bernoulli random variables. Therefore, with a probability of at least $1 - 2/n^6$,
\begin{align*}
    \widetilde{\mu}_3 & \leq \frac{nY}{2r_2} - \frac{n}{2\log n} \\
    & \leq \frac{n\left(\E[Y] + \sqrt{18\E[Y]\log n}\right)}{2r_2} - \frac{n}{2\log n} \\
    & \leq \mu(G[A]) + |R| + \sqrt{\frac{9n\cdot(\mu(G[A]) + |R|)\log n}{r_2}} - \frac{n}{2\log n} & (\text{Since } \frac{n\E[Y]}{2r_2}\leq \mu(G[A]) + |R|)\\
    & \leq \mu(G[A]) + |R| + \sqrt{\frac{n\cdot \left(\mu(G[A]) + |R| \right)}{32\log^2 n}}  - \frac{n}{2\log n} & (\text{Since } r_2 = 288 \log^3 n)\\
    & \leq \mu(G[A]), & (\text{Since } \mu(G[A])\leq n \text{ and }|R| \leq \frac{n}{4\log n})
\end{align*}
which implies that $\widetilde{\mu}_1 + \widetilde{\mu}_3 \leq \mu(G)$ with the same argument as proof of \Cref{lem:case1-approx}. Now we are ready to finish the proof since
\begin{align*}
    \widetilde{\mu}_1 + \widetilde{\mu}_3 \geq (1+c\alpha)\widetilde{\mu}_1 & \geq (1+c\alpha)\left((1 - \epsilon)\left(\frac{2}{3}-120\sqrt{\alpha}\right)\cdot\mu(G) - o(n)\right) \\
    & \geq (1+c\alpha)\cdot \left(\frac{2}{3} - 120\sqrt{\alpha} - \frac{2}{3}\epsilon\right)\cdot \mu(G) - o(n)\\
    & \geq \left(\frac{2}{3} - 120\sqrt{\alpha} - \frac{2}{3}\epsilon +\frac{2}{3}c\alpha - 120c\alpha\sqrt{\alpha} - \frac{2}{3}c\alpha \epsilon \right)\cdot\mu(G) - o(n) \\
    & \geq \left(\frac{2}{3} + \alpha \right)\mu(G) - o(n),
\end{align*}
\end{proof}

\subsubsection{Approximation analysis of Case 3: Almost every edge of maximum matching of $G[A]$ has at least one of its endpoints matched by almost all $M_{AB}^i$}

\begin{lemma}\label{lem:bound-on-A}
Before the start of the $i$-th iteration of \Cref{alg:matchingbetterthan23} for $i \in [1, T]$: if the algorithm has not terminated yet, then it holds that $\mu(G[A]) \geq \left(\frac{1}{4} - (i-1)\cdot(10^{20}\alpha)\right) \cdot \mu(G)$.
\end{lemma}
\begin{proof}
    In the beginning of the section we assumed that $\widetilde{\mu} \leq (\frac{2}{3} + \alpha)\cdot\mu(G)$. Note that $A$ contains vertices that are not matched by the matching of estimation $\widetilde{\mu}_1$. Therefore, before the first iteration of the algorithm, by \Cref{lem:tight-instance-edcs}, set $A$ must contain a matching of size larger than $(\frac{1}{3}-800\alpha)\cdot \mu(G) \geq \mu(G)/4$ where the inequality follows from our choice of $\alpha$.

    We use induction to prove the lemma. Suppose that the algorithm has not terminated before the $i$-th iteration for $i > 1$. Note that since the algorithm has not terminated in the previous iteration for any of the maximal matching $M_{AB}^i$, the total number of maximal matching edges between vertices of $A$ is at most $c\alpha \widetilde{\mu}_1 n^{\gamma/2}
    \leq c\alpha \mu(G)n^{\gamma/2}$, which implies that the average degree of a vertex in $G_{M_{AB}}[A]$ is at most
    \begin{align*}
        \frac{c\alpha\mu(G)n^{\gamma/2}}{|A|} & \leq \frac{c\alpha\mu(G)n^{\gamma/2}}{(\frac{1}{4} - (i-2)(10^5\alpha)) \cdot \mu(G)} & (\text{By induction hypothesis})\\
       & \leq \frac{c\alpha\mu(G)n^{\gamma/2}}{1/6\cdot \mu(G)} & (i \leq T \text{ and } 10^{20} T \alpha \leq \frac{1}{12})\\
       & = 6c\alpha n^{\gamma/2}.
    \end{align*}
    Let $Z$ be the number of edges of $G[A]$ that are in one of the maximal matchings. Since we sample $r_3 = 10n^{1 - \gamma/2}\log n$ vertices of $A$, we have
    \begin{align*}
        \E[Z] \leq 60 c \alpha n \log n.
    \end{align*}
    Since each vertex of $A$ have at most $n^{\gamma/2}$ neighbors in $A$ (because we have $k = n^{\gamma/2}$ maximal matchings), then by using Hoeffding’s inequality, we obtain
    \begin{align*}
        \Pr[|Z-\E[Z]| \geq 30 c \alpha n \log n] \leq 2 \exp\left(- \frac{2 \cdot(10n^{1-\gamma/2}\log n) \cdot (12 \alpha n^{\gamma /2})}{n^\gamma} \right) \leq \frac{1}{n^{10}}.
    \end{align*}
    
    Therefore, with high probability, there are at most $90 c \alpha n \log n$ edges of maximal matchings of sampled vertices with both endpoints in $A$ (*).

    Now we show that at most $10^{20} \alpha \mu(G)$ vertices of $A$ are removed in Line~\ref{Line:remove} (of Algorithm~\ref{alg:matchingbetterthan23}) in the previous iteration. For the sake of contradiction, assume that more than $10^{20} \alpha \mu(G)$ are removed in the previous iteration. Then, we have at least
    \begin{align*}
        10^{20} \alpha \mu(G) \cdot \eta & = 10^{20} \alpha \mu(G) \cdot \left( \frac{n \log n}{100 \widetilde{\mu}_1}\right)\\
        & \geq 10^{20} \alpha \mu(G) \cdot \left( \frac{n \log n}{100 \mu(G)}\right) & (\text{Since } \widetilde{\mu}_1 \leq \mu(G))\\
        & = 100 c  \alpha n \log n & (\text{Since } c = 10^{16})
    \end{align*}
    edges of maximal matching between vertices of $A$ since each removed vertex must have at least $\eta$ neighbors in the sampled vertices, which is a contradiction to (*). By induction, before iteration $i - 1$, we have $\mu(G[A]) \geq (\frac{1}{4} - (i-1)\cdot(10^{20}\alpha)) \cdot \mu(G)$. Also, at most $10^{20} \alpha \mu(G)$ vertices of $A$ can be removed in the previous step which implies the claimed bound for iteration $i$.
\end{proof}

\begin{corollary}\label{cor:size-of-a}
    Before the start of the $i$-th iteration of \Cref{alg:matchingbetterthan23} for $i \in [1, T]$: if the algorithm has not terminated yet, then it holds that $|A| \geq \frac{1}{5}\mu(G)$ 
\end{corollary}
\begin{proof}
    By \Cref{lem:bound-on-A}, we have that
    \begin{align*}
        \mu(G[A]) \geq (\frac{1}{4} - (i-2)(10^{20}\alpha)) \cdot \mu(G) \geq \frac{1}{5}\mu(G),
    \end{align*}
    where the last inequality follows from our choice of $\alpha$ and $T$. We finish the proof by the fact that the number of vertices of a graph is larger than its matching size.
\end{proof}

\begin{observation}
    Suppose that \Cref{alg:matchingbetterthan23} does not terminate in iteration $i$. Then, it holds $|B| \geq \mu(G)/5$ before iteration $i$.
\end{observation}

\begin{proof}
    Consider maximal matching $M_{AB}^1$. Since the algorithm has not terminated in iteration $i$, at least $\mu(G[A]) - 2c\alpha \widetilde{\mu}_1$ edges of $\mu(G[A])$ are blocked by maximal matching edges since the algorithm does not find more than $2c\alpha \widetilde{\mu}_1$ edges of matching $\mu(G[A])$ in case 1 and 2. This means that at least one endpoint of these edges is matched with a vertex in $B$. Thus, we have
    \begin{align*}
        |B| \geq \mu(G[A]) - 2c\alpha \widetilde{\mu}_1 \geq (\frac{1}{4} - (i-1)(10^{20} \alpha) - 8\alpha)\cdot\mu(G) \geq \frac{1}{5} \cdot \mu(G),
    \end{align*}
    where the second inequality follows from \Cref{lem:bound-on-A} and $\widetilde{\mu}_1 \leq \mu(G)$, and the last inequality follows by the choice of $T$, $\alpha$, and $c$.
\end{proof}

\begin{lemma}\label{lem:removing-b}
    Suppose that the \Cref{alg:matchingbetterthan23} does not terminate in iteration $i$. Then, at the end of this iteration, the algorithm removes at least $0.01 \cdot |B|$ vertices of $B$.
\end{lemma}

\begin{proof}
    Since the algorithm does not terminate in iteration $i$ for any of $k = n^{\gamma/2}$ maximal matchings, at least $(\mu(G[A]) - 2c\alpha\mu(G))\cdot n^{\gamma/2}$ edges of $\mu(G[A])$ are blocked by each of maximal matchings. This means that at least one endpoint of these edges is matched with a vertex in $B$. Thus, there are at least \begin{align*}
        (\mu(G[A]) - 2c\alpha\cdot \mu(G))\cdot n^{\gamma/2} \geq (\frac{1}{4} - (i-1)(10^{20} \alpha) - 2c\alpha)\cdot\mu(G) n^{\gamma /2} \geq \frac{1}{5} \cdot \mu(G)n^{\gamma/2}
    \end{align*}
    edges with one endpoint in $A$ and one endpoint in $B$ in all maximal matchings. Hence, the expected average degree of a vertex in $B$ in the subgraph $G_{M_{AB}}$ is at least
    \begin{align*}
        \frac{\frac{1}{5} \cdot \mu(G)n^{\gamma/2}}{|B|} \leq \frac{1}{5} \cdot n^{\gamma/2}, 
    \end{align*}
    where the inequality is because $|B| \leq \mu(G'') \leq \mu(G)$. Let $Z$ be the number of edges of $G_{M_{AB}}$ with one endpoint in $A$ that is among $r_3$ sampled vertices and one endpoint in $B$. Since we sample $10n^{1-\gamma /2} \log n$ vertices of $A$, we have
    \begin{align*}
        \E[Z] \geq 2n\log n.
    \end{align*}
    With the same argument as proof of \Cref{lem:bound-on-A}, each vertex of $A$ can have at most $n^{\gamma/2}$ neighbors in $B$, then by using Hoeffding's inequality, we get
    \begin{align*}
        \Pr[|Z-\E[Z]| \geq n \log n] \leq 2 \exp\left(- \frac{2 \cdot(10n^{1-\gamma/2}\log n) \cdot (
        \frac{1}{10} n^{\gamma /2})}{n^\gamma} \right) \leq \frac{1}{n^{10}},
    \end{align*}
    which implies that with high probability, there are at least $n\log n$ edges of maximal matchings with one endpoint in sampled $A$ vertices and one in $B$.

    Next, we prove for each vertex in $B$, the number of sampled $A$ vertices in $G_{M_{AB}}$ is small with high probability. Note that each vertex in $A$ is sampled with a probability of $\frac{10n^{1-\gamma/2}}{|A|}$. Moreover, each vertex in $B$ can have at most $n^{\gamma/2}$ neighbors in  $G_{M_{AB}}$. Let $W_v$ be the expected number of sampled $A$ vertices that are connected to $v$ in $G_{M_{AB}}$ for $v \in B$. Therefore,
    \begin{align*}
        \E[W_v] \leq \frac{10n \log n}{|A|} \leq \frac{50 n \log n}{\mu(G)}.
    \end{align*}
    where the last inequality holds because of \Cref{cor:size-of-a}. Then, by Chernoff bound
    \begin{align*}
        \Pr\left[|W_v - \E[W_v]| \geq \frac{25n\log n}{\mu(G)}\right] \leq 2\exp\left(-\frac{625n\log n}{150\mu(G)}\right) \leq \frac{2}{n^8}.
    \end{align*}
    Using a union bound on all vertices of $B$ implies that with high probability, each of the vertices of $B$ has at most $75n\log n / \mu(G)$ neighbors in $G_{M_{AB}}$ among sampled $A$ vertices. 

    For the sake of contradiction, suppose that the algorithm removes less than $0.01|B|$ vertices of $B$ at the end of the iteration. This implies that the maximum number of edges between vertices of $B$ and sampled $A$ vertices is
    \begin{align*}
        0.01|B| \cdot \left(\frac{75n\log n}{\mu(G)}\right) + 0.99|B| \cdot \left(\frac{n\log n}{100\widetilde{\mu}_1}\right) & \leq 0.01|B| \cdot \left(\frac{75n\log n}{\widetilde{\mu}_1}\right) + 0.99|B| \cdot \left(\frac{n\log n}{100\widetilde{\mu}_1}\right) \\
        &\leq   \frac{0.76|B|\cdot n\log n}{\widetilde{\mu}_1}\\
        & \leq 0.76 n \log n,
    \end{align*}
    where the first term is the total number of edges of removed vertices of $B$ and the second term is the total number of edges of vertices of $B$ that are not removed. Since we proved that with high probability, there are at least $n\log n$ edges, then the above upper bound on the number of edges is a contradiction that completes the proof.
\end{proof}

\subsection{Putting it all together: Beating two-thirds}

\begin{lemma}\label{lem:termination}
    If $\widetilde{\mu} \leq (\frac{2}{3} + \alpha)\cdot\mu(G)$, then the \Cref{alg:matchingbetterthan23} terminates in one of $T$ iterations.
\end{lemma}

\begin{proof}
    Suppose that the algorithm does not terminate in any of $T$ rounds. By \Cref{lem:removing-b}, in each iteration, the algorithm will remove $0.01 |B|$ vertices of $B$. Thus, after $T = 200$ iterations, none of the vertices of $B$ will remain in the graph. However, plugging $T =200$ in \Cref{lem:bound-on-A}, we have $\mu(G[A]) \geq \frac{1}{5}\mu(G)$. Therefore, either $\widetilde{\mu}_2 \geq c\alpha \widetilde{\mu}_1$ or $\widetilde{\mu}_3 \geq c\alpha \widetilde{\mu}_1$ in iteration $200$, since there is no vertices of $B$ to block the matching of $G[A]$.
\end{proof}

\begin{lemma}\label{lem:approx-ratio-better-32}
    The output of \Cref{alg:matchingbetterthan23} is a $(\frac{2}{3} + \alpha, o(n))$ estimate for maximum matching of $G$.
\end{lemma}
\begin{proof}
    If $\widetilde{\mu}$ is not a $(\frac{2}{3} + \alpha, o(n))$ estimate for maximum matching of $G$, then by \Cref{lem:termination}, the algorithm terminates in one of the rounds. Therefore, by \Cref{lem:case1-approx} and \Cref{lem:case2-approx}, the output is a $(\frac{2}{3} + \alpha, o(n))$ estimate for maximum matching of $G$.
\end{proof}

\apprxtheoremB*

\begin{proof}
    By \Cref{lem:approx-ratio-better-32}, we obtain the claimed approximation ratio. Also, since we choose $\epsilon$ to be smaller than $\gamma$ in \Cref{alg:matchingbetterthan23}, the running time is $\widetilde{O}_\epsilon(n^{2-\epsilon^3})$ by \Cref{lem:time-better-23}.
\end{proof}

\paragraph{Multiplicative Approximation for Adjacency List:} With the same argument as the multiplicative approximation for the adjacency list in the previous section, we can assume $\mu(G') = \Omega(n^{0.99})$. Hence, if $\widetilde{\mu} < (2/3 + \alpha)\mu(G)$, we have $\mu(G[A]) = \Omega(n^{0.99}) $ by \Cref{lem:tight-instance-edcs}. Therefore, with the exact same argument as previous section, using a standard Chernoff bound, for getting multiplicative approximation, we can use $r_1 = \widetilde{O}(n^{0.01})$ and $r_2 = \widetilde{O}(n^{0.01})$ samples for estimation of $\widetilde{\mu}_2$ and $\widetilde{\mu}_3$. Note that the running time of \Cref{clm:maximal-matching-aa-edges} and \Cref{clm:estimating-aa-missing-matching} will increase to $\widetilde{O}_\epsilon(n^{1+\epsilon + \gamma/2 + 0.01})$ and $\widetilde{O}_\epsilon(n^{2+\epsilon-\gamma/2 + 0.01})$, respectively because of the larger number of samples. However, if we choose $\gamma$ larger than $0.02 + 2\epsilon + 2\epsilon^3$, the running time of the algorithm will remain $\widetilde{O}(n^{2-\epsilon^3})$ by \Cref{lem:time-better-23}.

\apprxtheoremBlist*


\bibliographystyle{plainnat}
\bibliography{references}
	
\end{document}